\documentclass[a4paper,11pt,,english]{article}

\usepackage[T1]{fontenc}
\usepackage{hyperref}
\usepackage{verbatim}
\usepackage{amsfonts}
\usepackage{amsmath,amssymb,amsthm}
\usepackage{url}
\usepackage{fullpage}
\usepackage[usenames]{color}
\usepackage{times}
\usepackage{babel}
\usepackage{enumitem}

\newtheorem{theorem}{Theorem}[section]
\newtheorem{lemma}[theorem]{Lemma}
\newtheorem{proposition}[theorem]{Proposition}

\newtheorem{corollary}[theorem]{Corollary}

\newtheorem{claim}[theorem]{Claim}
\newtheorem{obs}[theorem]{Observation}
\newtheorem{definition}[theorem]{Definition}

\newcommand{\N}{{\mathbb N}}

\newcommand{\ignore}[1]{}

\newcommand{\dsum}{\displaystyle\sum}
\newcommand{\expect}{\mathbb{E}}

\newcommand{\eps}{\epsilon}

\newcommand{\mcA}{{\mathcal A}}
\newcommand{\mcLA}{{\mathcal L \mathcal A}}

\newcommand{\mcB}{{\mathcal B}}
\newcommand{\mcP}{{\mathcal P}}

\newcommand{\mcF}{{\mathcal F}}
\newcommand{\mcM}{{\mathcal M}}
\newcommand{\maj}{{\succeq}}

\begin{document}

\date{}

\title{Local Computation Mechanism Design}
\author{Avinatan Hassidim   \thanks{Department of Computer Science, Bar Ilan University. E-mail: {\tt avinatan@cs.biu.ac.il}. This research was supported by supported by ISF grant number 1241/12, and by BSF grant number 2012344. } \and {Yishay Mansour
\thanks{Blavatnik School of Computer Science, 
Tel Aviv University.
 E-mail: {\tt  mansour@tau.ac.il}. This
 research was supported in part by the Google Inter-university Center
 for Electronic Markets and Auctions, the Israeli Centers of Research
 Excellence program, the Israel Science Foundation, the United
 States-Israel Binational Science Foundation, and the Israeli
 Ministry of Science .}}
\and {Shai Vardi
\thanks{Blavatnik School of Computer Science, 
Tel Aviv University.
 E-mail: {\tt  shaivar1@post.tau.ac.il}. This research was supported in part by the Google Europe Fellowship in Game Theory.}
 }}




\maketitle
\begin{abstract}
We introduce the notion of \emph{local computation mechanism
design} - designing game theoretic mechanisms that run in
polylogarithmic time and space.
Local computation mechanisms reply to each query in polylogarithmic time and
space, and the replies to different queries are consistent with the same
global feasible solution.
When the mechanism employs payments, the computation of the payments is also done in
polylogarithmic time and space.
Furthermore, the mechanism needs to maintain incentive compatibility
with respect to the allocation and payments.



We present local computation mechanisms for a variety of classical
game-theoretical problems: (1) stable matching, (2) job scheduling,
(3) combinatorial auctions for unit-demand and $k$-minded bidders,
and (4) the housing allocation problem.

For stable matching, some of our techniques may have 
implications to the global (non-LCA) setting. Specifically, we show that  when the men's preference
lists are bounded, we can achieve an arbitrarily good approximation
to the stable matching within a fixed number of iterations of the
Gale-Shapley algorithm.
\end{abstract}

\section{Introduction}
\label{section:introduction}

Assume that we would like to design an auction for millions of
buyers and items. Alternatively, we have a cloud of hundreds of
thousands of computers on which we would like to schedule several
millions of jobs. In the not-so-distant past, these ideas would have
been unthinkable, but today, technological advances, especially the
Internet, have led us to the point where they are not only possible,
but necessary. One can easily conceive a cloud computation with
thousands of selfish computers, each one wanting to minimize its
work load. Alternatively, an ad-auction for millions of businesses
competing for advertising on millions of websites does not appear to
be a far away dream. In cases like these, the data sets on which we
need to work are so large, that polynomial-time tractability may not
be enough. Sometimes, even computing a solution in linear time may
be infeasible.  Often, however, only parts of the solution to a
problem are required at each point in time. In such cases, we can
use \emph{local computation algorithms} (LCAs).

Local computation algorithms, which were introduced by Rubinfeld et al. \cite{RTVX11}, consider the scenario in
which we need to be able to respond to queries (regarding a feasible
solution) quickly, but we never need the entire solution at once.
For example, in most auctions, this is a reasonable assumption. When queried, we
need to be able to tell each buyer which items she received and how
much to pay; for a given item we need to tell the seller to whom and when
to ship the item. There is no need to calculate the entire allocation and payment at
any specific time or to commit the entire solution to memory. Having
an LCA to such an auction would mean that we can reply to queries in
polylogarithmic time and only require polylogarithmic space.
Furthermore, if all of the items and buyers are queried, combining
the results will give us a complete solution that meets our
requirements.

The field of \emph{algorithmic mechanism design}
is an area at the intersection of economic game theory and algorithm
design, whose objective is to design mechanisms in decentralized
strategic environments. These mechanisms need to take into account
both the algorithmic efficiency considerations and the selfish
behavior of the participating agents.

In this paper we propose {\em local computation mechanism design},
which shares the motivations of both local computation algorithms
and algorithmic mechanism design. Our abstract model is the
following: We have a large data-set and a set of allowable queries.
Our goal is to implement each query locally, with polylogarithmic
time and space, while maintaining the incentives of participants. It
is worthwhile to give a few illustrative examples:
\begin{enumerate}[noitemsep, nolistsep]
\item Consider the problem of assigning doctor interns to hospitals internships, the classical motivation for stable matching.
We would like to be able to compute, for each doctor, her assigned hospital,
without performing the entire global computation.
\item  Consider a large distributed data center that has to assign jobs
to machines and elicits from each machine its speed. When queried on a job, we would like
to reply to which machine it is assigned, and when queried regarding a machine, we would like to reply with the set
of jobs that need to run on it. Again, we would like the computation to be
local, without constructing a global solution,
and still be able to ensure the machines have an incentive to report their speeds truthfully.
\item Consider a large auction (for instance, an ad auction platform). When queried regarding a bidder, we  would like
to compute the items she receives and her payment; when queried on
an item, we would like to compute which bidder won it, and it's price. Again, we would like guarantee that the bidders have
an incentive to report their preferences truthfully.
\end{enumerate}
\medskip

The following are our main contributions.
%
First, we formalize the notion of \emph{local computation mechanism
design}. A mechanism is \emph{local} if, for every query, it
calculates an allocation (and a payment) in polylogarithmic time and
space. Furthermore, the allocation must be consistent with some
global solution, and the payment must ensure truthfulness of the
agents. Second, we present local computation mechanisms for several
interesting problems, where our main result is an LCA for stable matching. Third, we use our techniques to show that 
in the general case when the men's lists
have bounded length (even in cases that do not admit an LCA), we can
find arbitrarily good matchings
(up to both additive and multiplicative constants)
by truncating the Gale-Shapley algorithm to a constant number of
rounds. 

We provide LCAs for the following problems:
\medskip
\paragraph{Stable matching}
In the \emph{stable matching} (or \emph{stable marriage}) problem,
introduced by Gale and Shapley \cite{GS62}, we would like to find a
\emph{stable} perfect matching between a group of $n$ men and a
group of $n$ women\footnote{A blocking pair is a man $m$ and a woman $w$ such that $m$ prefers $w$ to the woman he is matched to, and $w$ prefers $m$ to the 
man she is matched to. A matching is stable if there is no blocking pair.}.
We focus on the model introduced by Immorlica and Mahdian \cite{IM05}, in which the the women can have arbitrary preferences over the men, and the men have preference lists of length $k$ over the women, sampled uniformly at random.

Our main result is a local computation algorithm which matches all but an arbitrarily small fraction of the participants (this is often called an \emph{almost} stable matching; see, e.g., \cite{TL84,LZ03}). 
Furthermore, limited to the matched participants, the matching is stable.

\medskip
\paragraph{Scheduling on related machines}
In the \emph{makespan minimization} problem,  we want to schedule $n$ jobs on $m$
machines so as to minimize the maximal running time (makespan) of
the machines. This problem has many variations; we consider the scenario in which $m$  identical jobs need to be
allocated among $n$ related machines. The machines are strategic agents, whose private
information is their speed. We show:
\begin{enumerate}[noitemsep,nolistsep]
\item A local mechanism that is truthful in expectation for scheduling on related machines, which provides an $O(\log\log{n)}$-approximation to the optimal makespan.
\item A local mechanism that is universally truthful for the restricted case
(i.e., when each job can run on one of at most a constant number of predetermined machines), which provides an $O(\log\log{n})$-approximation to the optimal makespan.
\end{enumerate}

We also show some subtle and surprising results on the truthfulness of our algorithms.

\medskip
\paragraph{Matching combinatorial auctions} \emph{Combinatorial auctions} (CAs) are auctions in which buyers can
bid on bundles of items.
We consider the following scenario: $m$ items are to be auctioned
off to $n$ unit-demand  buyers, where each buyer is interested in a set of at
most $k$ items, sampled uniformly at random. We show universally  truthful local mechanisms for
the following variations, both of which provide a
$1/2$-approximation
to the optimal solution:
\begin{enumerate}[noitemsep,nolistsep]
\item When all buyers have an identical valuation for the items in their sets,
and the buyers' private information is the sets of items they are
interested in.
\item
When the sets are public knowledge, and the buyers' private
information is their valuation.
\end{enumerate}

We also show that there cannot exist an (exact) LCA for
\emph{maximum matching}, and therefore, there cannot exist any local
mechanism which computes the optimal solution.

\medskip
\paragraph{Combinatorial auctions with $k$-single minded bidders}
If each buyer is interested in a set of at most $k$ items, sampled uniformly at random, and has
private valuation for this set, we show a universally truthful local
mechanism which gives a $1/k$-approximation to the optimal
social welfare.

\medskip
\paragraph{Random Serial Dictatorship (RSD)}
We show that in
the housing allocation setting,
in which each agent is interested in a constant number of houses, sampled uniformly at random, the RSD algorithm is implementable as an LCA.


\subsection{Related Work }

\noindent{\em Local Computation Algorithms:} 
Rubinfeld et al. \cite{RTVX11}, showed how to transform distributed algorithms to
LCAs, and gave LCAs for several problems, including maximal
independent set and hypergraph $2$-coloring. 
Alon et al. \cite{ARVX11}, expanded the work of \cite{RTVX11} and gave better
space bounds for maximal independent set and hypergraph
$2$-coloring, using \emph{query trees}. Query trees were introduced
in the local setting by Nguyen and Onak \cite{NO08}: a random
permutation of the vertices is generated, and a sequential algorithm
is simulated on this order. The query tree represents the dependence
of each query on the results of previous queries. \cite{NO08} showed that  if the graph has a bounded degree, the query tree has a constant expected size.  \cite{ARVX11} showed that the query tree has polylogarithmic size with high probalility, and 
that the space required by the algorithm can be reduced by using a random seed to generate the ordering. 
 Mansour et al. \cite{MRVX12}, showed that the size of the query tree
can be bounded, with high probability, by $O(\log{n})$, and showed
how it is possible to transform many on-line algorithms to LCAs.
Using this technique, they showed LCAs for maximal matching and
several machine scheduling problems. Mansour and Vardi \cite{MV13},
showed an LCA which finds a $(1-\eps)$-approximation to the maximum
matching.


\smallskip
\noindent{\em Mechanism Design}:
We divide our attention between two types of mechanisms: with and
without payments. When the mechanism designer can incorporate
payments (to or from participating parties), these payments are
usually used to guarantee \emph{incentive compatibility} or
``truthfulness'' (\cite{MWG95}, \cite{NR99}): the agents are
rationally motivated to truthfully reveal their private information.
The mechanisms with payments discussed in this paper are all
randomized, and there are several quantifications of truthfulness
for such mechanisms in the literature, e.g., \cite{FKKV13,HKV13}. We
focus on the two most widely accepted (e.g., \cite{NR99,AT01,LOS02,DD09,DNS12}): \emph{truthfulness in expectation}, in which
the expectation of each agent's utility is maximized by being
truthful (where the expectation is taken over the coin flips of the
mechanism), and \emph{universal truthfulness}, where each agent's
utility is maximized by being truthful, regardless of the
realization of the randomness of the mechanism. When the mechanism
does not support payments, it is sometimes impossible to guarantee
truthfulness without crippling restrictions  to the mechanism
\cite{A50}. In such cases, it is common to look at (Bayesian)
\emph{ex-ante truthfulness}\footnote{Usually referred to in Economic
literature as \emph{ex-ante efficiency}.}, in which the expectation
of each agent's utility is maximized by being truthful (where the
expectation is taken over the prior distribution of the other
agents' private information).

Mechanisms are generally required to run in time (and space)
polynomial in the size of the input. In cases when the optimal
solution can be calculated in polynomial time and space, one can
apply the well-known VCG payments \cite{V61,C71,G73}.
Unfortunately, these payments can only be applied when the optimal
solution can be computed \cite{NR99,LOS02}, and in many cases, it
has been shown that computing an optimal solution is NP-hard
\cite{NR99,D11}. In these cases, we can only hope to design
algorithms that approximate the optimal solution in polynomial time.


\smallskip
\noindent{\em Additional Related Work:} Because of the large variety
of game-theoretic settings considered  in this paper, instead of
listing the entire glossary of related game-theoretic works here, we
provide a short  subsection dedicated to related work  pertaining to
each topic at the start of the relevant sections.

\subsection{Paper Organization}

In Section~\ref{section:preliminaries} we give some general definitions and
notation which we will use in the rest of the paper, and  present
our model for local computation game theoretic mechanisms. 
In Section~\ref{section:stable_matchings} we present our LCA for stable matching.  In Section~\ref{section:global_GS} we show some properties of the (global) Gale-Shapley algorithm that can be derived using our proof techniques.  In Section \ref{section:ex_ante}, we show that our LCA for stable matching is ex-ante truthful. In Section~\ref{section:scheduling} we present LCAs for machine scheduling; in Section~\ref{section:matching}, we give LCAs for combinatorial auctions with unit demand bidders, and prove that there cannot be an LCA for maximum matching; in Section~\ref{sec:SMB1} we extend the results of Section~\ref{section:matching} to combinatorial auctions with single-minded bidders. In Section~\ref{sec:rsd}, we give an LCA for random serial dictatorship, and finally, in Section~\ref{sec:conclusions}, we provide several interesting open questions for future research.

\section{Our Model and Preliminaries}
\label{section:preliminaries}


\subsection{Local Computation Algorithms}

The model we use is a generalization of the model of local computation algorithms (LCAs) introduced in
\cite{RTVX11}.\footnote{ Our model differs from the model of
\cite{RTVX11} in that their model requires that the LCA
\emph{always} obeys the time and space bounds, and returns an error
with some probability. It is easy to see that any algorithm which
conforms to our model can be modified to conform to the model of
\cite{RTVX11} by forcing it to return an error if the time or space
bound is violated (the other direction does not necessarily hold). Note however, that using this translation, a
truthful mechanism in our model would not necessarily translate to a
truthful mechanism in the model of \cite{RTVX11}.} A {\em $(t(n),
s(n), \delta(n))$-local computation algorithm} $\mcLA$ for a
computational problem is a (randomized) algorithm that receives an
input of size $n$, and a query $x$. Algorithm $\mcLA$ replies to
query $x$ in time $t(n)$ and uses at most $s(n)$ memory, with
probability at least $1-\delta(n)$. Furthermore, the replies to all
of the possible queries are consistent and combine to a feasible
solution to the problem. That is, the algorithm \emph{always}
replies correctly, but there is a $\delta(n)$ probability that the
time and/or space bounds will be violated. 
\subsection{Mechanism Design}

We use the standard notation of game theoretic mechanisms. There is
a set $N$ of $n$ rational  agents and a set $M$ of $m$ objects. In
some settings, e.g., the stable marriage setting, there are no
objects, only rational agents.  Each agent $i\in N$ has a valuation function $v_i$ that maps subsets $S\subseteq M$ of the items to non-negative numbers. The utilities of the agents are quasi-linear, namely, when agent $i$ receives subset $S$ of items and pays $p$, her utility is $u_i(S,p)= v_i(S)-p$. Agents are rational in the sense that they select actions to maximize their utility.  We would like to allocate items
to agents (or possibly agents to other agents), in order to meet
global goal, e.g., maximize the sum of the valuations of allocated
objects (see, e.g., \cite{NRTV07}).


A {\em mechanism with payments} $\mcM = (\mcA, \mcP)$ is composed of
an allocation function $\mcA$, which allocates items to agents, and
a payment scheme $\mcP$, which assigns each agent a payment. A
mechanism without payments consists only of an allocation function.
 Agents report their bids  to the mechanism.
Given the bids $b=(b_1,\ldots,b_n)$, the mechanism allocates  the item subset $\mcA_i(b)\subseteq M$ to agent $i$, and, if the mechanism is with payments,  charges her $\mcP_i(b)$; the utility of agent $i$ is $u_i(b)=v_i(\mcA_i(b)) - \mcP_i(b)$. 



A randomized mechanism is {\em universally truthful} if for every
agent $i$, for every random choice of the mechanism, reporting her
true private valuation maximizes her utility.
A randomized mechanism is {\em truthful in expectation}, if for
every agent $i$, reporting her true private valuation maximizes her
expected utility. That is, for all agents $i$, any bids $b_{-i}$ and
$b_i$, $ \expect[u_i(v_i, b_{-i})] \geq \expect[u_i(b_i, b_{-i})].$

We say that an allocation function $\mcA$ \emph{admits} a truthful payment scheme if there exists a payment scheme $\mcP$ such that the mechanism $\mcM = (\mcA, \mcP)$ is truthful.

A mechanism $\mcM = (\mcA, \mcP)$ fulfills \emph{voluntary participation} if, when the agent bids truthfully,
the utility of every agent is always
non-negative, regardless of the other agents' bids, i.e., for all agents $i$ and bids $b_{-i}$, $u_i(v_i, b_{-i})\geq 0\;$.

\subsection{Local Computation Mechanisms}
\label{subsection:model}

\begin{definition}
[Mechanisms without payments] We say that a mechanism  $\mcM$ is
\emph{$(t(n), s(n), \delta(n))$-local}  if its allocation function
is computed by a $(t(n), s(n), \delta(n))$-local computation
algorithm.
\end{definition}

\begin{definition}
[Mechanisms with payments] We say that a mechanism  $\mcM = (\mcA,
\mcP)$ is  \emph{$(t(n), s(n), \delta(n))$-local} if both the
allocation function $\mcA$ and the payment scheme $\mcP$ are
computed by $(t(n), s(n), \delta(n))$-local computation algorithms.
\end{definition}

In other words, given  a query $x$, $\mcA$ computes an allocation
and $\mcP$ computes a payment, and both run in time $t(n)$ and space
$s(n)$ with probability at least $1-\delta(n)$. Furthermore, the
replies  of $\mcA$ to all of the queries combine to a feasible
allocation.

A \emph{truthful local mechanism} $\mcM = (\mcA, \mcP)$ is a local mechanism which is also truthful.
%
Namely, each agent's dominant bid is her true valuation, regardless of the fact that the mechanism is local.\\

\section{Stable Matching}
\label{section:stable_matchings}

The \emph{stable matching} problem  is represented by a tuple
$A=(M,W, P)$, where $M$ is the set of men, $W$ is set of
women, and $P$ is the set of preference relations of the men and the
women: each man $m \in M$ has a preference relation over the women:
if $m$ prefers $w$ to $w'$ we denote this by $w
\overset{m}{\succ}w'$. Similarly there is a preference relationship
$\overset{w}{\succ}$ for each woman $w$.

\paragraph{Related work}
\emph{Stable matching} has been at the center of game-theoretic
research since the seminal paper of Gale ans Shapley \cite{GS62} (see,
e.g., \cite{R03} for an introduction and a summary of
many important results). Roth and Rothblum \cite{RR99} examined the
scenario in which the preference lists are of bounded length; in
most real-life scenarios, this is indeed the case. For example, a
medical student will not submit a preference list  for internship over all of the
hospitals in the United States. Furthermore, the mechanism designer
in most of these cases \emph{de facto} decides on the list length
because the mechanism will usually require the men to submit a list
of some predetermined length. It is  known that a linear number of
iterations of the Gale-Shapley algorithm is necessary to attain
stability \cite{GI89}, and several works address the situation when
we are interested in a sublinear number of queries, for example Feder et al.
\cite{FMP00} propose a parallel sublinear time algorithm for stable
matching. Unfortunately, it is not possible to convert their
algorithm to a local computation algorithm. Several experimental
works on parallel algorithms for the stable matching problem provide
evidence that after a constant number of rounds, the matching is
almost stable (e.g., \cite{TL84,Q85,LZ03}). 
Floreen et al. \cite{FKPS10} show that in the special case when both the mens' and
the womens' preference lists are bound by a constant, there exists a
distributed version of the Gale-Shapley algorithm, which can be run
for a constant number of rounds and finds an almost stable matching.

We examine the variant in which each man $m \in M$ is interested in
at most $k$ women, (and prefers to be unmatched than to be matched
to anyone not on their list; cf. \cite{RR99}). We limit our
attention to the setting in which the men's preference are assumed
to be uniformly distributed; cf. \cite{IM05,KP09}. 

\begin{definition}\label{def:stable}
A {\em matching} $H:M \rightarrow W \cup \{\perp\}$ is a function
which is injective over $W$ (but it is possible that for several men
$m_i$, $H(m_i)=\perp$ - these men are said to be
\emph{unmatched}).\\
A matching $H$ is said to be \emph{stable} if for every man $m$ such
that $H(m) = w$ (possibly $w=\perp$), and every $w'$ such that $w'
\overset{m}{\succ} w$, then $H^{-1}(w') \overset{w'}{\succ} m$. A
couple $(m,w')$ is said to be \emph{unstable} if $\exists w \in W
\cup \{\perp\}: H(m)=w$, $w' \overset{m}{\succ} w$, and for $m'=
H^{-1}(w')$, $m \overset{w'}{\succ} m'$ (possibly $m'=\perp$).\\
The stable matching problem $(M,W,P)$ where each man has a
preference list of length $k$ and each woman is chosen uniformly at
random is called  \emph{$k$-uniform}.
%
\end{definition}

The Gale-Shapley algorithm 
finds a stable matching in the $k$-uniform setting (e.g., \cite{GS85}). To
ensure the locality of our algorithm, we allow our mechanism to
"disqualify" men, in which case they remain unmatched, but are
unable to contest the matching. We try, however, to keep the number
of disqualified men to a minimum. Our main result is the following.
\begin{theorem}
\label{thm:GS} Let $A=(M,W, P)$ be a stable matching problem,
$|M|=|W|= n$, in the $k$-uniform setting. Then there is an
$(O(\log{n}), O(\log{n}), 1/n)$-local computation mechanism for $A$
in which at most an $O(1/k)$ fraction of the men remain unmatched.
\end{theorem}


We begin by describing a non-local algorithm, \textsc{AbridgedGS},
and then show how to simulate it locally by a local algorithm,
\textsc{LocalAGS}.

\subsection{\textsc{AbridgedGS}}

Let  \textsc{AbridgedGS} be the Gale-Shapley men's courtship
algorithm, where, in addition to the preference lists being of
constant length, the algorithm is stopped after $\ell$ rounds. That
is, in each round, each unassigned man goes to the highest ranked
woman who has not yet rejected him. Each woman then keeps the man
she prefers out of the men who approached her, and rejects the rest.
This continues until the $\ell^{th}$ round, and the men who were
rejected on the $\ell^{th}$ round are disqualified. (Note that the
men who were rejected $k$ times are left unmatched as well, but they are not said to be disqualified.) We
simulate \textsc{AbridgedGS} using an LCA.

\subsection{\textsc{LocalAGS} - an LCA implementation of \textsc{AbridgedGS}}

Consider a graph which represents the problem, where the men and the women are represented by vertices, and an edge exists between
two vertices if and only if they are on each other's list. Define the \emph{distance} between two agents to be the length
of the shortest path between them in the graph.
Define the $d$-neighborhood of a person $v$ to be everyone at a
distance at most $d$ from $v$, denoted $N_d(v)$

Assume we are queried on a specific man, $m_1$. We simulate
\textsc{AbridgedGS} locally as follows: Let the number of rounds be
$\ell=2k^2$ (see Lemma \ref{lemma:GSshai}). We look at the
$2\ell$-neighborhood of $m_1$. For each man $m_i \in
N_{2\ell}(m_1)$, we simulate round $1$. Then, for each man $m_i \in
N_{2\ell-2}(m_1)$, we simulate round $2$. And so on, until for $m_i
\in N_2(m_1)$, that is, $m_1$ and his closest male neighbors, we
simulate round $\ell$. We return the woman to whom $m_1$ is paired,
``unassigned'' if he was rejected by $k$ women, and ``disqualified''
if he was rejected by a woman in round $\ell$. We denote this
algorithm by \textsc{LocalAGS}.

\begin{claim}
For any two men, $m_i$ and $m_j$, whose distance from each other is
greater than $2\ell$, $m_i$'s actions cannot affect $m_j$ if Algorithm \textsc{AbridgedGS}  terminates after $\ell$ rounds.
\end{claim}

\begin{proof}
The proof is by induction.
For $\ell=1$, let $w_1$ be $m_j$'s first choice. Only men for whom
$w_1$ is their first choice can affect $m_j$, and these are a subset of the men at
distance $2$ from $m_j$. For the inductive step, assume that the
claim holds for $\ell-1$. Assume that there is a man $m_i$ whose
actions can affect $m_j$ within $\ell$ rounds, who is at a distance
of at least $2\ell+2$ from $m_j$. From the inductive claim, none of
$m_i$'s actions can affect any of $m_j$'s neighbors within $\ell-1$
rounds. As their actions in round $\ell-1$ (or any previous round)
will not be affected by $m_i$, and they are the only ones who can
affect $m_j$ in round $\ell$, it follows that $m_i$ cannot affect
$m_j$ within $\ell$ rounds.
\end{proof}

\begin{corollary}
 The query of a man $m$'s status in round $\ell$
only needs to consider men at distance at most $2\ell$ from $m$.
\end{corollary}

\begin{lemma}\label{lemma:run}
The running time and space of algorithm \textsc{LocalAGS} is
$O(\log{n})$ with probability at least $1-1/n^2$.
\end{lemma}

We prove the following claim, which immediately implies Lemma
\ref{lemma:run}.

\begin{claim}\label{claim:run}
For sufficiently large $n$, for any integer $i>0$, there exists a
constant $c_i$ such $\Pr[|N_{i}(v)| \leq c_i \log{n}] \geq 1-1/n^2$ .
\end{claim}

\begin{proof}
Let $N^i_v$ be the random variable representing the number of
vertices in the $i$-neighborhood of vertex $v$. Note that as the
degree of a women $v$ is distributed binomially $N^1_v \sim
B(n,k/n)$ and $\expect[N^1_v] = k$. We prove by induction that
$\Pr[N_v^i \leq c_i(\log{n})] \geq 1-\frac{i}{n^3}$.

For the base, $i=1$, if $v$ is a man, $N_v^1=k$. If $v$ is a woman,
we employ the Chernoff bound with $\lambda>2e-1$:\footnote{This
bound is reached by substituting $\lambda \geq 2e-1$ into the
standard Chernoff bound $\Pr[X>(1+\lambda)\mu]\leq \left(
\frac{e^{\lambda}}{(1+\lambda)^{1+\lambda}}\right)^{\mu}$ }
$\Pr[N_v^1> (1+\lambda)k]<2^{-k\lambda}$. Therefore, for $c_1=4$ and
for $n\geq 2^{k}$,

\begin{equation*}
 \Pr[N_v^1 > c_1\log{n}]\leq 2^{-c_1\log{n}+k}<\frac{2^k}{n^{c_1}}\leq\frac{1}{n^3},
\end{equation*}

Assuming that the claim holds for all integers smaller than $i$,  we
show that it holds for $i$. If the outermost vertices of the
neighborhood are men, then $N_v^i \leq kN_v^{i-1}$ and we can take
$c_i = kc_{i-1}$. Otherwise, we use the law of total probability.
\begin{align}
\Pr[N_v^i> c_i\log{n}] =& \Pr[N_v^j> c_i\log{n}|N_v^{i-1}\leq c_{i-1}\log{n}]\Pr[N_v^{i-1}\leq c_{i-1}\log{n}]\notag\\
&  +  \Pr[N_v^j> c_i\log{n}|N_v^{j-1}>c_{i-1}\log{n}]\Pr[N_v^{j-1}>c_{i-1}\log{n}]\notag\\
\leq& \Pr[N_v^j > c_i\log{n}|N_v^{j-1} \leq c_{i-1}\log{n}] + \Pr[N_v^{j-1}>c_{i-1}\log{n}]\notag\\
\leq& \Pr[N_v^j > c_i\log{n}|N_v^{j-1} \leq c_{i-1}\log{n}] + \frac{i-1}{n^3}.\notag
\end{align}
where the last inequality uses the inductive hypothesis.


The probability that the degree of any node $u$ is exactly $z$ is at
most
\begin{equation*}
\Pr[deg(u) = z] \leq  {n \choose z}\left(\frac{k}{n}\right)^{z} \leq
\left(\frac{ek}{ z}\right)^{z},
\end{equation*}
using the inequality ${n \choose k} \leq
\left(\frac{ne}{k}\right)^k$.
Hence, for $z\geq e^2 k$ we have that $\Pr[deg(u) = z] \leq e^{-z}$
and $\Pr[deg(u) = z] \leq e^{-\widehat{z}}$, where
$\widehat{z}=\max\{0,z-e^2k\}$.
%

We like to bound the probability that $N^{i}_v$ is larger than
$c_{i}\log n$ although $N^{i-1}_v$ is less than $c_{i-1}\log n$.
We define a new random variable $\widehat{N}^i_v$ as follows. Let
$y\leq c_{i-1}\log n$ be the number of nodes at distance $i-1$ from
$v$ and let $z=(z_1, z_2,\ldots, z_{y})$ be their degrees. We define
the truncated degrees as  $\widehat{z}=\{\widehat{z}_1,
\widehat{z}_2,\ldots ,\widehat{z}_{y}\}$ such that
$\widehat{z}_j=\max\{0,z_j-e^2k\}$. The value of $\widehat{N}^i_v$
is the sum of the truncated degrees at distance $i-1$ from $v$,
i.e., $\widehat{N}^i_v=\sum_{i=1}^{y}\widehat{z}_i$. Clearly $ N^i_v
\leq\widehat{N}^i_v+e^2k y\leq \widehat{N}^i_v+c_{i-1}e^2k\log n $.
Therefore it is sufficient to bound $\widehat{N}^i_v$.

Let $\widehat{x} = \sum_{i=1}^y \widehat{z}_i$. The probability that the truncated degrees of the vertices at
distance $i-1$ are exactly $\widehat{z}=(\widehat{z}_1,
\widehat{z}_2,\ldots ,\widehat{z}_{y})$ is at most $\prod_{i=1}^{y}
e^{- \widehat{z}_i} =e^{- \widehat{x}}$. There are
$\binom{\widehat{x} + y}{y}$ vectors $\widehat{z}$ that can realize
$\widehat{x}$. We  bound $\Pr[\widehat{N}_v^i =
\widehat{x}|N_v^{i-1} \leq y]$, for $\widehat{x}\geq 7y$
as follows:
\begin{align*}
\Pr[\widehat{N}_v^i = \widehat{x}|N_v^{i-1} \leq y] \leq
\binom{\widehat{x} + y}{y} e^{- \widehat{x}}
 \leq \left(\frac{e\cdot(\widehat{x}+\widehat{x}/7)}{\widehat{x}/7} \right)^{\widehat{x}/7} e^{- \widehat{x}}=  e^{-(1-(1+\ln(8))/7 )\widehat{x}}
\leq e^{-\widehat{x}/2},
\end{align*}
It follows that
\begin{align}
\Pr[\widehat{N}_v^i \geq 7y |N_v^{i-1} \leq y] &\leq
\dsum_{\widehat{x}=7 y}^{\infty} e^{- \widehat{x}/2}=
\frac{e^{-7y/2}}{1-e^{-1/2}}
\notag \leq e^{- y} \leq 1/n^3,\label{eq:beta2}
\end{align}
which follows since $c_{i-1}\geq 3$.
%
%
Therefore for $c_i= (e^2k+7)c_{i-1}\leq(16 k)^i$ we have,
\begin{equation*}
\Pr[N_v^i> c_i\log{n}] \leq \frac{1}{n^3}+ \frac{i-1}{n^3}=\frac{i}{n^3}.
\end{equation*}
\end{proof}

Claim \ref{claim:run} implies Algorithm \textsc{LocalAGS} makes
$O(\log{n})$ queries with probability at least $1/n^2$, and so Lemma
\ref{lemma:run} follows.

\begin{lemma}
\label{lemma:GSshai}
In Algorithm \textsc{LocalAGS}, setting $\ell = 2k^2$  ensures at most $4n/k$ men remain unmatched with probability at least $1-\frac{1}{n^2}$.
\end{lemma}

\emph{Note}: This implies that  the mechanism can limit the number
of unmatched pairs with a high degree of certainty by specifying the
length of the list.

Lemmas \ref{lemma:run} and \ref{lemma:GSshai} imply Theorem \ref{thm:GS}.
We prove Lemma   \ref{lemma:GSshai} in the following subsection; the proof follows from  Claims \ref{claim:short} and \ref{claim:limit}.

\subsection{Bounding the number of men removed}

To prove Lemma \ref{lemma:GSshai}, we bound the number of men remaining unassigned due to the fact that
the lists are short, and the number of men disqualified due to the
number of rounds being bounded.

\subsubsection{Removal due to short lists}

We bound the number of  unpaired women as a result of the fact that
the lists are short, using the principle of deferred decisions
\cite{K76}\footnote{Instead of ``deciding'' on the preference lists
in advance, each man chooses the $(i+1)^{th}$ woman on his list only
if he is rejected from the $i^{th}$ - this mechanism is known to be
equivalent to the mechanism which we use (see, e.g.,  \cite{K76}).}.
Note that the number of unpaired women equals the number of unpaired
men.

\begin{lemma}
\label{lemma:short} In the $k$-uniform setting, the Gale-Shapley
algorithm results in at most $\frac{2n}{k}$ men being unpaired, with
probability at least $1-1/2n^2$.
\end{lemma}

\begin{proof}
Consider the following stochastic process:  in the first round, each
of the $n$ men chooses a woman independently  and uniformly at
random. For each consecutive round, for each woman that has been
chosen by at least one man, one of the men remains married to her
(arbitrarily chosen), and the others remain single and choose again.
This process repeats for $k$ rounds. This is modeled by the
functions $f^t:M\rightarrow W$, where in round $t$, $f^t$ maps each
single man to a woman uniformly at random and each married man to
the same woman.
Let $\mcF$ be the set of all possible allocation functions $f:M \rightarrow W$.

For ease of analysis, we assume that each man can choose the same woman again, as the number of free women in this case is an upper bound to the number of free women in the system where he can not.
Note that the number of unmatched men after $t$ rounds is identical to the number of unmatched men  after $t$ rounds of the Gale-Shapley process. This stochastic process, however, terminates after $k$ rounds, whereas the Gale-Shapley process can continue. As the number of matched men can only increase when more rounds are added, the number of unmatched men created by the process is an upper bound to the number of unmatched men created by the Gale-Shapley process.

Let $X_j^t$ be the indicator variable which is $1$ if woman $j$ is unassigned at the end of round $t$. Let $X^t = \sum_{j=1}^n X_j^t$.
The following claim implies Lemma  \ref{lemma:short}.
\begin{claim}
\label{claim:short}
For any constant $t$, $\Pr[X^t > \frac{2n}{t}] \leq \frac{t}{n^3}$.
\end{claim}
\begin{proof}
The proof is by induction. The base of the induction, $t=1$, is immediate.
For the inductive step, assume that after round $t$, $X^{t}=n/z$ (for some $z>0$).
In round $t+1$,
$\expect[X^{t+1}|X^{t}=\frac{n}{z}] = \frac{n}{z}(1-1/n)^{n/z}$.
 By the inductive hypothesis
\begin{equation}
 \Pr \left[X^{t} > \frac{2n}{t}\right] \leq \frac{t}{n^3}. \label{eq:greater}
\end{equation}

  For the rest of the proof, assume $X^t \leq \frac{2n}{t}$, and fix $X^t$ to be some such value. We get

\begin{align}
\expect[X^{t+1}|X^t \leq \frac{2n}{t}] \leq \frac{2n}{t}(1-1/n)^{2n/t}
<\frac{n}{t/2 \cdot e^{2/t}}
<\frac{2n}{t+2},\label{1}
\end{align}
using $e^x>1+x$.

Order the women arbitrarily, ($W = \{1,2,\ldots n\})$, and let $W_i = \{1, 2, \ldots i\}$. For $h \in F$, define the martingale
\begin{equation*}
Y^{t+1}_i(h) = \expect[X^{t+1}(f^{t+1})|f^{t+1}(j)=h(j) \text{ for all } j \in W_i],
\end{equation*}
where $X^{t+1}(f^{t+1})$ is the realization of $X^{t+1}$ given that the allocation vector is $f^{t+1}$.
Note that $Y^{t+1}_0(h)$ is the expected value of $X^{t+1}$ over all possible functions $f^{t+1}$; that is, the expected number of unmatched women after $t+1$ rounds. $Y^{t+1}_n(h)$ is the number of unmatched women after $t+1$ rounds when the allocation function is $h$. $X^{t+1}$ satisfies the Lipschitz condition, because if $h$ and $h'$ only differ on the allocation of one man, $|X^{t+1}(h) - X^{t+1}(h')| \leq 1$. Therefore
\begin{equation*}
|Y^{t+1}_{i+1}(h) - Y^{t+1}_i(h)| \leq 1,
\end{equation*}
(see, e.g. \cite{AS08}), and so we can apply Azuma's inequality, from which we get
\begin{equation*}
\Pr[|X^{t+1}-\expect[X^{t+1}]|> \lambda\sqrt{n}] < 2e^{-\lambda^2/2}.
\end{equation*}
Setting $\lambda = \frac{2\sqrt{n}}{(t+1)(t+2)}$, we have that

\begin{equation*}
\Pr \left[|X^{t+1}-\expect[X^{t+1}]|> \frac{2n}{(t+1)(t+2)}\right] <
2e^{-n/( 81t^4)},
\end{equation*}
for constant $t\geq 2$.

Therefore, since we assume that $\expect[X^{t+1}] <\frac{2n}{t+2}$,
\begin{equation}
\Pr\left[X^{t+1}> \frac{2n}{t+1}\right] < 2e^{-n/(81t^4)}
<\frac{1}{n^3}\label{2}.
\end{equation}

Note that
\begin{align*}
\Pr\left[X^{t+1}> \frac{2n}{t+1}\right] &= \Pr\left[X^{t+1}> \frac{2n}{t+1} | X^t \leq \frac{2n}{t}\right] \Pr\left[X^t \leq \frac{2n}{t}\right]\\
& +\Pr\left[X^{t+1}> \frac{2n}{t+1} | X^t > \frac{2n}{t}\right]\Pr\left[X^t > \frac{2n}{t}\right]
\end{align*}
From the induction hypothesis and Equations \eqref{eq:greater}, \eqref{1} and \eqref{2}, using the union bound,
\begin{equation*}
\Pr\left[X^{t+1}> \frac{2n}{t+1}\right]\leq\frac{t}{n^3}+\frac{1}{n^3}\leq \frac{t+1}{n^3}.
\end{equation*}
\end{proof}
The stochastic process for which Claim \ref{claim:short}  holds ends
at least as early as the Gale-Shapley algorithm with short lists;
therefore Claim \ref{claim:short} implies Lemma \ref{lemma:short}
\end{proof}

\subsubsection{Removal due to the number of rounds being limited}

Because we stop the \textsc{LocalAGS} algorithm after a constant
($\ell$) number of rounds, it is possible that some men who ``should
have been'' matched are disqualified because they were rejected by
their $i^{th}$ choice in round $\ell$ ($i<k$). We show that this
number cannot be very large.

 Let $R_i$ denote the number of men rejected in round $i \geq 1$.
\begin{obs}
\label{obs:Xr}
$R_i$ is monotone decreasing in $i$.
\end{obs}
\begin{claim}
\label{claim:limit}
The number of men rejected in round $i$ is at most $n\frac{k}{i}$.
\end{claim}
\begin{proof}
As each man  can be rejected at most $k$ times, the total number of
rejections possible is $kn$. The number of men who can be rejected
in round $i$ is at most
\begin{align}
{R_{i} \leq kn - \dsum_{j=1}^{i-1} R_j}  
\;\;\;\Rightarrow \;\;\;   {R_{i} \leq kn - (i-1) R_{i}} 
\;\;\;\Rightarrow\;\;\; {R_{i} \leq n\frac{k}{i}}\;\;\;
\label{eqXr}
\end{align}
Where Inequality \ref{eqXr} is due to Observation \ref{obs:Xr}.
\end{proof}

\section{Some general properties of the Gale-Shapley algorithm}
\label{section:global_GS}
We use the results and ideas of Section \ref{section:stable_matchings} to prove some interesting features of the (general) Gale-Shapley stable matching algorithm, when the mens' lists are of length at most $k$. (These results immediately extend to our local version of the algorithm, \textsc{LocalAGS}.) Note that the proof of Claim \ref{claim:limit} makes no assumption on how the men's selection is made, and therefore, Claim \ref{claim:limit} implies that as long as each man's list is bounded by $k$, if we run the Gale-Shapley for $\ell$ rounds, at most $\frac{nk}{\ell}$ men will be rejected in that round. This immediately gives us an additive approximation bound for the algorithm if we stop after $\ell$ rounds.
\begin{corollary}[to Claim \ref{claim:limit}]
Assume that the output of the Gale-Shapley algorithm on a stable matching problem, where the preference lists of the men are of length at most $k$, is a matching of size $M^*$. Then, stopping the Gale Shapley algorithm after $\ell$ rounds will result in a matching of size at least $M^* - \frac{nk}{\ell}$.
\end{corollary} 
 We would like to also provide a multiplicative bound.
Henceforth, we assume that the mens' list length is bounded by $k$, but make no other assumptions.
For each round $i$, let $M_i$ be the size of the current matching; let $D_i$ be the number of men who have already approached all $k$ women on their list and have been rejected by all of them; let $C_i$ be the number of men who were rejected by women in round $i$, but have approached fewer than $k$ women so far; recall that $R_i$ is the number of men rejected in round $i$. Denote  the size of the matching returned by the un-truncated Gale-Shapley algorithm by $M^*$.

\begin{claim}
\label{claim:ck}
$C_{k+1} \leq kM^*$.
\end{claim}
\begin{proof}
 Note that $R_i = C_i +D_i-D_{i-1}$. For $i<k, D_i=0$. As $M_i$ is monotonically increasing, $\forall i\leq k, R_i\geq n-M^*$.
$$\dsum_{i=1}^{k} R_i \geq kn-kM^*.$$
Hence,
$$C_{k+1} \leq  kn -\dsum_{i=1}^k R_i \leq kM^*. $$
\end{proof}

\begin{corollary}
For every $\eps>0$, there exists a constant $\ell>0$ such that $C_{\ell} \leq \eps M^*$.
\end{corollary}
\begin{proof}
Denote the maximum number of rejections possible from round $i$ onwards by $L_i$. 
Clearly, 
$$L_i \leq k(M_i+C_i) \leq k(M^* + C_i).$$
For all $i$ such that $C_{i} \geq \eps M^*$, we have
 $$L_i \leq \left(1+\frac{1}{\eps}\right)k C_i.$$
 
Therefore, from Claim \ref{claim:ck},
\begin{equation*}
L_{k+1} \leq \left(1+\frac{1}{\eps}\right)k^2 M^*.
\end{equation*}
Putting everything together, we have,
\begin{align}
L_{i+1}& \leq  L_i - C_i \notag \\ 
\Rightarrow L_{i+1} & \leq L_i\left(1-\frac{1}{k(1+\frac{1}{\eps})}\right) \notag \\
\Rightarrow L_{k+i+1} & \leq L_k\left(1-\frac{1}{k(1+\frac{1}{\eps})} \right)^i \notag\\
  & \leq  \left(1+\frac{1}{\eps}\right)k^2M^*\left(1-\frac{1}{k(1+\frac{1}{\eps})} \right)^i \notag\\ 
  & \leq 2k^2M^*e^{-\frac{i}{k(1+\frac{1}{\eps})}}, \notag
\end{align}

Taking $i=k(1+\frac{1}{\eps}) \log{\frac{2k^2}{\eps}}$ gives $C_{k+i+1}\leq L_{k+i+1} \leq \eps M^*$.
\end{proof}
This gives us,

\begin{theorem}
Consider a stable matching problem. Let each man's list be bounded by $k$. Denote the size of the stable matching returned by the Gale-Shapley algorithm by $M^*$. Then, if the process is stopped after $O(\frac{k}{\eps} \log{\frac{k}{\eps}})$ rounds, the matching returned  is at most a $(1+\eps)$-approximation to $M^*$, and has at most $\eps M^*$ unstable couples. 
\end{theorem}

\begin{corollary}
If both men and women have lists of length at most $k$, then for any $\eps$ there is an $(O(1), O(1), 0)$-LCA for stable matching which returns a matching that is at most a $(1 + \eps)$-approximation to the matching returned by the Gale-Shapley algorithm, and with at most an $\eps$-fraction of the edges being unstable.
\end{corollary}

\section{Ex-ante truthfulness of LocalAGS}
\label{section:ex_ante}
It is known that the Gale-Shapley men's courtship algorithm is strategy-proof for the men but not for the women (see, e.g., \cite {MSZ13}). Unfortunately, Algorithm \textsc{LocalAGS} is neither strategy-proof for the men nor for the women: for the women, this follows immediately from the fact that the men's courtship algorithm is not strategy-proof; for the men - a man who was rejected in round $\ell$ might prefer to not declare the woman from which he was rejected in that round. It is possible, though, to show that \textsc{LocalAGS} is  ex-ante truthful.
The men  have a preference relation over the women (and women over the men). However, it is not clear how to calculate the expected utility in this case, so for the purposes of the proof,  assume  the men have a utility function $u:W\rightarrow \mathbb{R}$. The utility of the men for women not among their first $k$ preferences is $0$. Note that  \textsc{LocalAGS} has no access to the utilities themselves, but only to the preferences, therefore the only way for a man to manipulate the algorithm is by misrepresenting his preference vector.

Assume that man $m$'s real utility function is $t_m$, but he declares $b_m$. Denote the set of all the possible preference vectors of the other players by $\mcB_{-m}$. Let $b_{-m} \in \mcB_{-m}$ be a single realization of the preference vectors of the other players. Let $b=(b_m, b_{-m})$. The outcome of instance $b$  is deterministically determined; denote the utility of $m$ for instance $b$ by $u_m(b)$. ($u_m(b) = t_m(w)$ if $m$ is paired with woman $w$ and $0$ otherwise.)  Denote the expected utility of $m$  when he bids $b_m$  by $\overline{u_m}(b_m) = \dsum_{b_{-m} \in \mcB_{-m}} u_m(b)\cdot Pr[b_{-m}]$.

\begin{definition}
A transformation of a utility function by swapping the utilities of two women is called a \emph{swap}. 
\end{definition}
\begin{obs}
It is possible to convert any preference vector $v$ to any other preference vector $v'$  by a finite series of swaps.
\end{obs}
This implies that given that a utilities function has a fixed list of outcomes, it is possible to reach any permutation of the utilities from any other by a finite series of swaps.
\begin{claim}
\label{claim:mentruth}
Algorithm \textsc{LocalAGS} is ex-ante truthful for the men.
\end{claim}
\begin{proof}
Assume that man $m$ gains by declaring his utility function to be $b_m$, over his true function, $t_m$, where $b_m$ is a permutation of $t_m$. That is, $\overline{u_m}(b_m)>\overline{u_m}(t_m)$. (Recall that $m$ can make no better manipulations.)


Take any legal chain of swaps from $t_m$ to $b_m$: $<t_m=u^1, u^2, \ldots u^i, \ldots, u^n=b_m>$. There must be two consecutive functions, $u^i$ and $u^{i+1}$ on the chain, for which 

\begin{equation}
\label{eqman1}
\overline{u_m}(u^{i+1})>\overline{u_m}(u^{i}),
\end{equation}
 from the transitivity of the relation ``$>$''. 
 Functions $u^i$ and $u^{i+1}$ differ in the utilities of two women, say $w_x$ and $w_y$: $u^i(w_x) = u^{i+1}(w_y)$, $u^i(w_y) = u^{i+1}(w_x)$.
 
 Denote by $\mcB'_{-m}$ the  set of all preferences which is reached by taking $\mcB_{-m}$ and interchanging women $w_x$ and $w_y$  wherever they appear. Likewise the preference vectors of women $w_x$ and $w_y$ are interchanged. 
  Note that $\mcB_{-m} = \mcB'_{-m}$, because they are both uniform distributions over all possible preference vectors. 
 By symmetry, it must hold that 
\begin{equation*}
\overline{u_m}(u^{i})>\overline{u_m}(u^{i+1}),
\end{equation*}
in contradiction to Equation \eqref{eqman1}.
\end{proof}

Similarly, \begin{claim}
\label{claim:womentruth}
Algorithm \textsc{LocalAGS} is ex-ante truthful for the women.
\end{claim}
The proof is similar to the proof of Claim \ref{claim:mentruth}, and is omitted.

\section{Local machine scheduling} \label{section:scheduling}

In this section we consider the following job scheduling setting. There is a set
$\mathcal{M}$ of  $n$ machines (or ``bins'') and a set $\mathcal{J}$
of $m$ uniform jobs (or ``balls''). Each machine $i\in\mathcal{M}$ has an
associated capacity $c_i$ (also referred to as its
``speed''). We assume that the capacities are positive
integers. Given that $h_i$ jobs are allocated to machine $i$, its {\em load} is
$\ell_i=h_i/c_i$.  ($h_i$ is referred to as the
{\em height} of machine $i$.) The {\em utility} of machine $i$ is quasi-linear,
namely, when it has load $\ell_i$ and receives payment $p_i$ then its utility is
$u_i(\ell_i,p_i)=p_i-\ell_i$.


The {\em makespan} of an allocation is
$\max_i\{\ell_i\}=\max_i\{h_i/c_i\}$. In our setting, the players
are the machines and their private information is their true
capacities. Each machine $i$ submits a \emph{bid} $b_i$ (which
represents its capacity). Our goal is to elicit from the machines
the information about their capacities, in order to minimize the
makespan of the resulting allocation.
%


For any allocation algorithm $\mcA$, define $\mcA(b)$ to be the
allocation vector, which, given bid vector $b$,  assigns each job $j\in\mathcal{J}$ to a
unique machine $i\in\mathcal{M}$. Let $\mcA^j(b)$ be the machine to which job $j$ is allocated in $\mcA$. When the bids $b_{-i}$ are fixed, we sometimes omit them from the notation for clarity.

\paragraph{Related work}
Azar et al., \cite{ABK+99}  proposed the $\textsc{Greedy}[d]$
algorithm in the online setting, where $m$ balls need to be
allocated to $n$ bins, with the objective of minimizing the
makespan: each ball chooses, uniformly at random, $d$ bins, and
allocates itself to the least loaded bin among its $d$ choices at
the time of its arrival.  They also showed that the
maximal load is $\Theta(m/n) + (1+o(1)) \ln\ln{n}/\ln{d}$. A large
volume of work has been devoted to variations on this problem, such
as adding weights to the  balls  \cite{TW07};
and variations on the algorithm, such as the non-uniform ball
placement strategies of V{\"o}cking \cite{Vocking03}. Of particular
relevance to this work is the case of non-uniform bins:  Berenbrink
et al. \cite{BBFN14}, showed that in this case the maximum load can
also be bounded by $\Theta(m/n) + O(\ln\ln{n})$.
Hochbaum and Shmoys \cite{HS88} showed a PTAS for scheduling on
related machines. Lenstra et al. \cite{LST87}, presented a
$2$-approximation algorithm for scheduling on unrelated machines and
showed that the optimal allocation is not approximable to within
$\frac{3}{2}-\epsilon$. The problem of finding a truthful mechanism
for scheduling (on unrelated machines) was introduced by  Nisan and
Ronen,\cite{NR99}, who showed an $m$-approximation to the problem,
and a lower bound of $2$. Archer and Tardos \cite{AT01} were the
first to tackle the related machine case; they showed a randomized
$3$-approximation polynomial algorithm and a polynomial pricing
scheme to derive a mechanism that is truthful in expectation. Since
then, much work has gone into finding mechanisms with improved
approximation ratios, until Christodoulou and Kov{\'a}cs \cite{CK10}
recently settled the problem by showing a deterministic PTAS, and a
corresponding mechanism that is deterministically truthful. Babaioff et al. \cite{BKS10}, showed how to transform any monotone allocation rule for single parameter agents to a truthful-in-expectation mechanism.

\begin{definition}(Monotonicity)
\label{defn1} A randomized allocation function $\mcA$ is
\emph{monotone in expectation} if for any machine $i$, and any bids
$b_{-i}$ of the other machines, the expected load of machine $i$,
$E[\ell_i(b_i, b_{-i})]$, is a  non-decreasing function of $b_i$.\\
A randomized allocation function $\mcA$ is \emph{universally
monotone } if for any machine $i$, and any bids $b_{-i}$ of the
other machines, the load of machine $i$, $\ell_i(b_i, b_{-i})$, is
a non-decreasing function of $b_i$ for any realization of the
randomization of the allocation function.
\end{definition}
%
%
Given an allocation function $\mcA$, we would like to provide a
payment scheme $\mcP$ to ensure that our mechanism $\mcM =
(\mcA,\mcP)$ is truthful. It is known that a necessary and sufficient condition is
that the allocation function $\mcA$ is monotone.

\begin{theorem} \cite{Mye81,AT01})
\label{thm:AT} The allocation  algorithm $\mcA$ admits a payment
scheme $\mcP$ such that the mechanism $\mcM = (\mcA,\mcP)$ is
truthful-in-expectation (universally truthful) if and only if $\mcA$
is monotone in expectation (universally monotone).
\end{theorem}

We differentiate between two settings.  The
{\em standard} setting (cf. \cite{BBFN14,Wie07}) is a slight variation on the basic power-of-$d$ choices setting proposed in \cite{ABK+99}, for some constant $d\geq 2$. For each job $j$, the mechanism chooses a subset $M_j \subseteq \mathcal{M}$, $|M_j|=d$ of machines that the job can be allocated to. The probability that machine $i \in M_j$ is proportional to $b_i$. 

 In the {\em restricted} setting (cf. \cite{ANR95}), each job can  be allocated to a subset of at most $d$ machines, where the subsets $M_j$ are given as an input to the
allocation algorithm. The restricted setting models the case when the jobs have different requirements, and there is only a small subset of machines that can run each job.


Mansour et at. \cite{MRVX12}, showed that it is possible to transform any on-line algorithm on a graph of bounded degree (or whose degree is distributed binomially) to an LCA. The idea behind the reduction is simple: generate a random permutation on the vertices and simulate the on-line algorithm on this permutation. They show that, with high probability, this results in at most $O(\log{n})$ queries. We require the following theorem from \cite{MRVX12}.
\begin{theorem} \cite{MRVX12} \label{thm:local_bb}
Consider a generic on-line algorithm $\mathcal{LB}$ which requires
constant time per query, for $n$ balls and $m$ bins, where $n=cm$
for some constant $c>0$.
There exists an $(O(\log^4{n}), $ 
$O(\log^3{n}), 1/n)$-local
computation algorithm which, on query of a (ball) vertex $v \in V$,
allocates $v$ a (bin) vertex $u \in U$, such that the resulting
allocation is identical to that of $\mathcal{LB}$.\footnote{We need
an assumption that each ball restricted to a constant number of
machines, either explicitly or implicitly through the algorithm
$\mathcal{LB}$.}
\end{theorem}

\subsection{A truthful in expectation mechanism for the standard setting}
\label{sec:lb1}

For the presentation of the allocation algorithm, we regard the allocation field as consisting of slots of size $1$. Each machine $i$'s bid
$b_i$ represents the number of slots it ``owns''.
Given the bids $b$ of the machines, let $B=\sum_{i=1}^n b_i$.
Thus, we can treat this as $B$ slots of size $1$ each, where each machine $i$ {\em owns} $b_i$ slots.
In the allocation, we will allocate jobs to slots. When a job $j$ is allocated to a specific slot, this means that
 the machine that owns the slot receives $j$.
We provide the following simple on-line allocation algorithm $\mcA_{SLMS}$ (cf. \cite{BBFN14}):
\begin{enumerate}[noitemsep,nolistsep]
\item Choose for  job $j$ a subset $M_j$ of $d$ slots out of $B$,
where each slot has equal probability. (Note that $M_j$ may include different slots owned by the same machine.)
\item Given $M_j$, job $j$ is allocated to the least loaded slot
in $M_j$ (breaking ties uniformly at random). Slots are treated as being independent of their machines. That is, it is possible that if a job chooses two slots $a$ and $b$, which belong to machines $A$ and $B$, $a$ has fewer jobs than $b$, but $B$ is more loaded than $A$, in terms of the average of the loads of its slots.
\end{enumerate}
\emph{Note}: Although it may not be possible to compute $B$ and $M_j$ locally exactly, it has been shown in that an approximate calculation suffices (e.g., \cite{BCM03, Wie07}). The reader is invited to peruse the references herein for a more in-depth discussion.  
\begin{lemma}
\label{lemma:bbmon} The randomized allocation function $\mcA_{SLMS}$ is
monotone in expectation.
\end{lemma}
%

\begin{proof}
Let $B=\sum_{i'} b_{i'}$ and $B_{-i}=\sum_{i'\neq i} b_{i'}$. Since
all the slots are identical, by symmetry the expected number of jobs allocated
to each slot is exactly $m/B$. Since machine $i$ owns $b_i$ slots,
the expected height of machine $i$ is
\[
\expect[h_i(b_i)]= \frac{b_i}{B_{-i}+b_i}m,
\]
which is monotone increasing in $b_i$ (for $b_i, B_{-i} \geq 0$).
\end{proof}

From Theorem \ref{thm:AT}, we immediately get:
\begin{lemma}
\label{lemma:bbtie} The randomized allocation function $\mcA_{SLMS}$
admits a payment scheme $\mcP_{SLMS}$ such that the  mechanism  $\mcM_{SLMS} = (\mcA_{SLMS},\mcP_{SLMS})$  is
truthful in expectation.
\end{lemma}

It is interesting to note that the above algorithm does not admit a
universally truthful mechanism. To show this, we prove a slightly stronger claim: that the \textsc{Greedy} algorithm of \cite{ABK+99}, (in which  each job chooses $d$ machines at random, and is allocated to the least loaded among them, post placement\footnote{That is, if bin $A$ has capacity $4$ and height $2$ and bin $B$ has capacity $8$ and height $5$, the job will go to machine $B$.}, breaking ties arbitrarily), does not admit a universally truthful mechanism. 
\begin{claim}
\label{claim:gen1}
Algorithm \textsc{Greedy} is not universally monotone.
\end{claim}
\begin{proof}

Assume we have $4$ machines: $A$, $B$, $C$, and $D$, with bids $4$, $4$, $8$ and $1$ respectively.
The first $2$ jobs choose machines $A$ and $D$ (which we abbreviate to $AD$),
the next $2$ jobs choose $BD$,
and the next $6$ jobs choose $CD$.
After these $10$ jobs, the heights of the machines are $(2,2,6,0)$ (recall that the Greedy algorithm allocates according to the \emph{post-placement} load).
The $11$th job chooses $AB$, and the $12$th job chooses $AC$.
As ties are broken at random, assume machine $A$
receives job $11$.
Machine $C$ then receives job $12$, making the capacities $(3,2,7,0)$.

Now assume machine $C$ bids $9$, and the choices of the first $10$
jobs and the $12$th job remain the same, but because $C$ bid higher,
now the $11$th job chooses $C$ instead of $A$ (so job $11$ chooses
$BC$). Now machine $B$ receives the $11$th job and machine $A$
receives the $12$th job, making the capacities  $(3,3,6,0)$. Machine
$C$ received less jobs although it was bidding higher!
%
\end{proof}

\begin{corollary}
\label{corr:gen1}
Algorithm $\mcA_{SLMS}$ is not universally monotone.
\end{corollary}

By Theorem \ref{thm:local_bb}, the allocation function
$\mcA_{SLMS}$  can be transformed to
a $(O(\log^4{n})$, $O(\log^3{n})$, $1/n)$ LCA. We would now like to show a payment scheme $\mcP_{SLMS}$ such that the mechanism  $\mcM_{SLMS} =
(\mcA_{SLMS}$,$\mcP_{SLMS})$ is a local mechanism. (We overload the notation, letting $\mcA_{SLMS}$ represent both the on-line allocation algorithm and its respective LCA.)
%
We therefore need to show a payment scheme which can be implemented in an
LCA and guarantees truthfulness. Our payment schemes are similar in idea to the payments schemes of \cite{APTT03} and \cite{BBFN14}.

\begin{lemma}
\label{lemma:tie_payments}
There exists a randomized local payment scheme $\mcP_{SLMS}$ such that the mechanism $\mcM_{SLMS} = (\mcA_{SLMS}, \mcP_{SLMS})$ is truthful in expectation. Furthermore, if all the bids are bounded by a polylogarithmic function, there exists a {\em deterministic} local payment scheme $\mcP_{SLMS}$ such that the mechanism $\mcM_{SLMS} = (\mcA_{SLMS}, \mcP_{SLMS})$ is truthful in expectation.
\end{lemma}

\begin{proof}
Archer and Tardos \cite{AT01} showed that the following payment
scheme makes for a truthful mechanism fulfilling voluntary
participation. For bid $b_i$:
\begin{align*}
p_i(b_i, b_{-i}) = b_i h_i(b_i, b_{-i}) + \int_0^{b_i}
h_i(x,b_{-i})dx\;.
\end{align*}
This has to be the expected payment, and we can (deterministically)
take $h_i(b)$ to be the expected height of machine $i$ when it bids
$b_i$. Since $E[h_i(b)]=m \frac{b_i}{B_{-i}+b_i}$, where
$B_{-i}=\sum_{i'\neq i} b_{i'}$, the payment is of the form:
\begin{align*}
p_i(b_i, b_{-i}) &=  m \frac{b_i^2}{B_{-i}+b_i} +
\dsum_{x=0}^{b_i}h_i(x,b_{-i}) \\
&= m\frac{b_i^2}{B_{-i}+b_i} +
m\dsum_{x=0}^{b_i}\frac{x}{B_{-i}+x}.
\end{align*}
Our remaining challenge is to compute the payment in polylogarithmic time and space.

If $b_i$ is bounded by some polylogarithmic function of $n$, we can
calculate the sum in polylogarithmic time and space, and take this
to be the payment. If, however, $b_i$ is larger, it is not apparent
how to calculate this sum  in a straightforward fashion.\footnote{We can
approximate the payment up to an additive factor of $1/B_{-i}$ using
$p_i(b_i, b_{-i})\approx m\frac{b_i^2}{B_{-i}+b_i} + m(b_i
-B_{-i}\ln(1+b_i/B_{-i}))$, but this might change the incentives of
the machines and is not guaranteed to be truthful.} We provide the
following (randomized) payment scheme:

Choose, uniformly at random, $k \in [1,b_i]$, and take the payment 
to be
\begin{equation*}
m \frac{b_i^2}{B} +   m b_i\cdot \frac{k}{B_{-i}+k}
\end{equation*}
This gives the correct expected payment, and takes $O(1)$ time.
%
\end{proof}
 \cite{BBFN14}, showed that $\mcA_{SLMS}$ provides an $O(\log\log{n})$ approximation to the optimal makespan. Therefore, by Theorem \ref{thm:local_bb}, the LCA of $\mcA_{SLMS}$ provides the same approximation ratio.
Combining Lemma \ref{lemma:bbtie}, and Lemma
\ref{lemma:tie_payments}, we state our main result for the standard
setting:

\begin{theorem}\label{thm:lb1}

There exists an $(O(\log^4{n}),O(\log^3{n}), 1/n)$- local mechanism to scheduling on
related machines in the standard setting that is truthful in expectation, and provides an $O(\log{\log{n}})$-approximation to the makespan.
\end{theorem}


\subsection{A universally truthful mechanism for the restricted setting}
\label{sec:lb2}

In the restricted setting, each job can only be allocated to one of a set $M_j \subseteq \mathcal{M}$ of $d$
machines. As opposed to the standard setting, $M_j$ is not selected by the
mechanism, but is part of the input. We assume 
that the allocation is random (this is necessary to bound the query tree size), and the probability of machine $i$ to be in $M_j$ is  proportional
to its capacity $c_i$ (necessary to guarantee the approximation ratio of the allocation algorithm $\mcA_{RLMS}$)\footnote{Although this assumption is somewhat strong for a theoretical discussion, from a practical viewpoint it usually holds that a machines' capacity is somewhat proportional to the number of jobs that is usually scheduled on it.}. The latter requirement can be relaxed slightly, but
for clarity of the proofs, we will assume that it holds exactly.
Furthermore, we assume that the capacity of each machine is not too large, that is, each machine has capacity at most $O(\log{n})$. 
We restrict our attention to the case where we would like to allocate  $m=\Omega(C)$ jobs where $C$ is the total capacity of the machines. This is usually considered to be the worst case scenario (see e.g., \cite{ABK+99,BBFN14,Wie07}).
We define the following (on-line) algorithm $\mcA_{RLMS}$ for assigning
jobs to machines as follows.
Initially, we select a permutation $\pi$ of the machines, for
tie-breaking.
Place job $t$ in the machine $i\in M_j$ for which the {\em post-placement
load}, $lp^{t+1}_i(b_i) = \lfloor \frac{h_i^{t}(b_i)+1}{b_i}\rfloor $
is smallest, breaking ties according to $\pi$. The following claim shows why it is necessary to take the floor of
the load, as the simple Greedy algorithm does not admit a universally
truthful mechanism in this case.

\begin{claim}The \textsc{Greedy} algorithm (unmodified) is not universally monotone in the restricted case.
\label{claim:app2}
\end{claim}

\begin{proof}
Assume we have $3$ machines $A,B,C$, with bids $(4,8,36)$
respectively, and a tie-breaking permutation: $A<B<C$ (Jobs always
prefer machine $A$ to machines $B$ and $C$, and machine $B$ to
machine $C$). The allocation until time $t$ has caused the heights
to be $(1,3,18)$. The next job's choices are machines $A$ and $B$
(which we abbreviate to $AB$), and the following two jobs' choices
are $BC$ and $AB$ respectively. The first job is allocated to $A$
(since the post-placement loads on $A$ and $B$ are $2/4$ and $4/8$
respectively, and $2/4 = 4/8$, we use the tie-breaking rule). The
second job is allocated to $B$ ($4/8<19/36$) and the third job to
$B$ ($5/8>3/4$). The heights of the machines are now  $(2,5,18)$.

Now let $B$ declare its capacity to be $9$, and assume that until
time $t$, there is no difference in the allocation. The loads at
time $t$ in this case are: $1/4, 3/9,  18/36$. The jobs' choices are
part of the input to the mechanism, so are unaffected by the bids,
and remain $AB, BC, AB$. The first job is allocated  to $B$
($2/4>4/9$), the second job to $C$ ($19/36 < 20/36 = 5/9$), and the
third job to $A$ ($2/4<5/9$). The heights of the machines are now
$(2,4,19)$. Thus,  $B$ gets fewer jobs after bidding higher.
\end{proof}





\begin{theorem}
For any permutation $\pi$ of the machines and any job arrival order, the
allocation function $\mcA_{RLMS}$ is universally monotone increasing in the machines'
bids.
\end{theorem}

From Definition \ref{defn1}, it suffices to prove the following lemma:

\begin{lemma}
\label{lemma:ai} For any machine $i$, fixing $b_{-i}$, for any
$b'_i>b_i$, we have that $h_i(\mcA_{RLMS}(b'_i, b_{-i})) $ $ \geq$ $h_i(\mcA_{RLMS}( b_i,
b_{-i}))$.
\end{lemma}

To prove Lemma \ref{lemma:ai}, define $D^t(k, b'_i, b_i)$ to be the difference in the number of
jobs allocated to machine $k$ between $\mcA_{RLMS}(b'_i)$ and
$\mcA_{RLMS}(b_i)$ up to and including time $t$. We abbreviate this to $D^t(k)$ when
$b'_i$ and $b_i$ are clear from the context. (If machine $k$ received less jobs, then
$D^t(k)$ is negative.) We say that machine $k$ \emph{steals} a job
from machine $l$ at time $t$ if $\mcA_{RLMS}^t(b_i) = l$ and
$\mcA_{RLMS}^t(b'_i) = k$. We will show that the only machine for which
$D^t(k)$ can be positive at some time $t$ is machine $i$, therefore, as
$\sum_{j=1}^n D^t(j)=0$, we have that $D^t(i)$ can never be negative.


\begin{proposition}
\label{prop:main} For any machine $i$, fixing $b_{-i}$, if
$b'_i>b_i$ then at all times $t$, for any machine $k \neq i$,
$D^t(k) \leq 0$.
\end{proposition}

Informally, Proposition \ref{prop:main} says that if bin $i$ claims its capacity is larger than it actually is, no bin except for $i$ can receive more balls. The following corollary follows immediately from Proposition
\ref{prop:main}, and implies Lemma \ref{lemma:ai}.
\begin{corollary}
\label{corr:main} For any machine $i$, fixing $b_{-i}$, if
$b'_i>b_i$ then at all times $t$, $D^t(i) \geq 0$.
\end{corollary}

Before proving Proposition~\ref{prop:main}, we first will make the
following simple observation
\begin{obs}
\label{obs:easy} For any machine $k$, if $D^t(k) \leq 0$ then
$lp^t_k(b'_i) \leq lp^t_k(b_i)$.
\end{obs}

\begin{proof}
For $k \neq i$, as $k$'s bid is the same in both allocations, if it
received less jobs by time $t$ in $\mcA_{RLMS}(b_i)$ then
the claim follows. If $k=i$, the claim
follows since $b'_i>b_i$.
\end{proof}

We now prove Proposition \ref{prop:main}:
\begin{proof}
The proof is by induction on $t$. At $t=1$, $D^1(k) = 0$ for every
$k$.

Assume the proposition is true for times $t= 1, \ldots , \tau-1$. We show it holds for $t=\tau$, by contradiction. Assume that we have a machine $k\neq i$ such that $D^\tau(k)>0$. At
time $\tau-1$, for all $k \neq i$, by the induction hypothesis, it
holds that $D^{\tau-1}(k) \leq 0$. The only way that $D^{\tau}(k)>0$
is if machine $k$ has $D^{\tau -1}(k) = 0$ and at time $\tau$ steals
a job. Assume first that machine $k$ steals a job from machine $l
\neq i$. This means that in $\mcA_{RLMS}(b_i)$, machine $l$ received job
$\tau$, therefore
\begin{equation}
lp^{\tau}_l(b_i) \leq lp^{\tau}_k(b_i). \label{eqn:perm1}
\end{equation}
By Observation \ref{obs:easy}, $lp^{\tau}_l(b'_i) \leq
lp^{\tau}_l(b_i)$, and so
\begin{center}
$lp^{\tau}_l(b'_i) \leq lp^{\tau}_l(b_i) \leq lp^{\tau}_k(b_i) =
lp^{\tau}_k(b'_i)$.
\end{center}
If machine $k$ steals job $\tau$ from machine $l$, then  $lp^{\tau}_k(b'_i)
\leq lp^{\tau}_l(b'_i)$. This is in contradiction to Equation
(\ref{eqn:perm1}) because there cannot be an equality both here and
in Equation (\ref{eqn:perm1}), as the tie-breaking permutation $\pi$
is fixed. More precisely, if $lp^{\tau}_l(b'_i) = lp^{\tau}_l(b_i) = lp^{\tau}_k(b_i) =
lp^{\tau}_k(b'_i)$, then job $\tau$ will be allocated to the same machine in $b_i$ and $b'_i$, according to the permutation $\pi$.

Therefore, machine $k$ must steal job $\tau$ from machine $i$, which
gives us
\begin{equation}
lp^{\tau}_i(b_i)  \leq lp^{\tau}_k(b_i)  = lp^{\tau}_k(b'_i) \leq
lp^{\tau}_i(b'_i). \label{eq1}
\end{equation}
The first inequality is due to the fact that machine $i$ receives
job $\tau$ in $\mcA_{RLMS}(b_i)$. The equality is due to the fact that
$D^{\tau-1}(k)=0$, and the second inequality is because machine $k$
receives job $\tau$ in $\mcA_{RLMS}(b'_i)$. And so,
\begin{equation}
lp^{\tau}_i(b_i) < lp^{\tau}_i(b'_i), \label{eq2}
\end{equation}
because one of the inequalities in Equation (\ref{eq1}) must be
strict, as the tie-breaking permutation $\pi$ is fixed.

Assume that the last time before $\tau$ that machine $i$ stole a job
is time $\rho$, and label by $z$ the machine that $i$ stole from at that time. We
have
\begin{equation*}
lp^{\rho}_i(b'_i)  \leq lp^{\rho}_z(b'_i)  \leq lp^{\rho}_z(b_i) \leq lp^{\rho}_i(b_i). \label{eq3}
\end{equation*}
The first inequality is because machine $i$ received job $\rho$ in
$\mcA_{RLMS}(b'_i)$. The middle inequality is because $D^{\rho}(z) \leq
0$. The last inequality is because machine $z$ received job $\rho$
in $\mcA_{RLMS}(b_i)$. Again, at least one inequality must be strict,
giving
\begin{equation*}
lp^{\rho}_i(b'_i)  <  lp^{\rho}_i(b_i),
\end{equation*}
which implies, for all $\alpha \geq 0$,
\begin{equation}
 \left\lfloor \frac{h_i^{\rho}(b'_i)+ \alpha + 1}{b'_i} \right\rfloor \leq  \left\lfloor \frac{h_i^{\rho}(b_i)+ \alpha}{b_i} \right\rfloor \label{eq4},
\end{equation}
since $b'_i>b_i \geq 1$.

Because job $\rho$ was the last job that machine $i$ stole, it
received at least as many jobs between $\rho$ and $\tau$ in
$\mcA_{RLMS}(b_i)$ as in $\mcA_{RLMS}(b'_i)$. Label the number of jobs $i$
received between $\rho$ and $\tau$ (including $\rho$ but excluding
$\tau$) in $\mcA_{RLMS}(b_i)$ by $\beta$ and in $\mcA_{RLMS}(b'_i)$ by
$\beta^*$.

\begin{obs}
\label{obs:beta}
$\beta^* \leq \beta+1$.
\end{obs}

\begin{proof}
Machine $i$ received at least as many jobs in $\mcA_{RLMS}(b_i)$ as in
$\mcA_{RLMS}(b'_i)$ after $\rho$. This must be true because $\rho$ was
the last time machine $i$ stole a job. However, machine $i$ received
the job at time $\rho$ in $\mcA_{RLMS}(b'_i)$ but not in $\mcA_{RLMS}(b_i)$,
and so we cannot claim that $\beta^* \leq \beta$, but only that
$\beta^* \leq \beta+1$.
\end{proof}
{\em Proof of Proposition \ref{prop:main} continued.}
From the definition of $lp$ and equation (\ref{eq4}), we get:
\begin{align}
lp^{\tau}_i(b'_i)  &= \left\lfloor \frac{h_i^{\tau}(b'_i)+1}{b'_i}\right\rfloor \notag\\
& = \left\lfloor \frac{h_i^{\rho}(b'_i)+\beta^*+ 1}{b'_i}\right\rfloor \label{eq41} \\
&\leq  \left\lfloor \frac{h_i^{\rho}(b'_i)+\beta+ 2}{b'_i}\right\rfloor \label{eq42}\\
&\leq  \left\lfloor \frac{h_i^{\rho}(b_i)+\beta+1}{b_i}\right\rfloor \label{eq43}\\
& = lp^{\tau}_i(b_i).\label{eq44}
\end{align}
Equality (\ref{eq41}) stems from the definition of $\beta^*$, Inequality (\ref{eq42}) is due to Observation \ref{obs:beta}, Inequality (\ref{eq43}) is due to Equation (\ref{eq4}), and Equality (\ref{eq43}) is from the definition of $\beta$.

This is in contradiction to Equation (\ref{eq2}), and therefore
$D^\tau(k)\leq 0$. This concludes the proof of the proposition.
%
\end{proof}

\begin{lemma}
\label{lemma:loglog}
The allocation algorithm $\mcA_{RLMS}$ provides an $O(\log{\log{n}})$-approximation to the optimal allocation.
\end{lemma}
The proof is similar to the proof for the unmodified Greedy algorithm in the case of non-uniform bins of 
\cite{BBFN14}. We provide it in Appendix~\ref{appendix:lglgn} for completeness.

\begin{lemma}
\label{lemma:ut_payments} There exists a local payment scheme
$\mcP_{RLMS}$ such that the mechanism $\mcM_2 = (\mcA_2, \mcP_2)$ is
universally truthful.
\end{lemma}

\begin{proof}
Having shown that $\mcA_{RLMS}$ is universally monotone, we can use the
payment scheme of  \cite{AT01}:
\begin{equation}
p_i(b_i, b_{-i}) = b_i h_i(b_i, b_{-i}) + \sum_{x=0}^{b_i} h_i(x,b_{-i}) \label {eqn:AT4}
\end{equation}
Unfortunately, the height of machine $i$ is not an easily computable function of $i$'s bid. We therefore need to explicitly compute the value of $h_i(x,b_{-i})$ for every $x \in [0,b_i]$.
 That is, we need to run the allocation algorithm $\mcA_{RLMS}$ again with each possible bid $x \in [0,b_i]$. As the running time of $\mcA_{RLMS}$ is $O(\log{n})$, and $b_i = O(\log{n})$, this will take at most
 $O(\log^2{n})$. Note that as the time to compute the permutation is $O(\log^3{n})$, computing the payment this way will not affect the asymptotic running time of the mechanism.
\end{proof}

We conclude:

\begin{theorem}\label{thm:lb1}

There exists an $(O(\log^4{n}),O(\log^3{n}), 1/n)$- local  mechanism to scheduling on
related machines in the restricted setting that is universally truthful and gives an $O(\log{\log{n}})$-approximation to the makespan.
\end{theorem}

\section{Local auctions for unit demand buyers}
\label{section:matching}

We propose local truthful mechanisms for auctions with unit-demand
buyers, in which each buyer is interested in at most $k$ items, and
each item is desired by at most a polylogarithmic number of buyers.
First, we tackle the case where all buyers have the same valuation
for the items in their sets, and the buyers' private information is
their sets. Then we examine the case in which the buyers' sets are
public knowledge and the buyers' private information is their
valuations for their items, with the restriction that buyers have
the same valuation for all items in their set.

\paragraph{Related work} \emph{Combinatorial
auctions} are an extremely well-studied problem in algorithmic game theory. The general premise is the following: we wish to allocate $m$ goods to $n$ players, who have valuations for subsets of goods, with the goal of maximizing the social welfare. The
general problem, where each player may have an an arbitrary
valuation for each subset of the goods is known to be $NP$-hard;
indeed, even approximating the optimal solution for single-minded
bidders to within $\sqrt{m}(1-\epsilon)$ is $NP$-hard \cite{LOS02}.
Therefore, in order to obtain useful approximation algorithms, we
must relax some of our demands. One such relaxation is limiting
ourselves to identical items. An example is the case of $k$-minded
bidders: There are $m$ identical indivisible goods, and $n$ buyers
with $k$ valuations each - each buyer $i$ has a valuation $v_i(j)$
for obtaining $j$ items (where $j$ is between $1$ and $k$).
 \cite{DN10} gave a PTAS for this problem, and
showed that (under certain restrictions), obtaining an FPTAS is
$NP$-hard.  \cite{KV12} gave universally truthful mechanisms for combinatorial auctions in an on-line model.

We use the following theorem from \cite{MRVX12} for finding maximal matchings in undirected graphs.

\begin{theorem} \cite{MRVX12}
\label{thm:maximal} Let $G=(V,E)$ be an undirected graph with $n$
vertices and maximum degree $d$. Then there is an $(O(\log^4{n}),
$ $O(\log^3{n}), $ $1/n)$ - local computation algorithm for
maximal matching.
\end{theorem}
As in the case of load balancing, the idea behind Theorem \ref{thm:maximal} is that one can simulate the well-known Greedy on-line algorithm for maximal matching, using a random permutation of the edges.

\subsection{Unit-demand buyers with uniform value}
\label{sec:udfv}

We first consider the following scenario. We have a set
$\mathcal{I}$ of $n$  unit-demand  buyers, and a set $\mathcal{J}$
of $m$ indivisible items. There is a fixed, identical value for all
items, which we normalize to $1$. Each buyer $i$ is interested in a
set $J_i$ of at most $k$ items (where $k$ is a constant).
We can treat this auction as a graph
$G=(V,E)$, in which $V=\mathcal{J} \cup \mathcal{I}$, and $E = \{
(i,j):i \in\mathcal{I},j \in J_i\}$. The value of a subset $S$ to
buyer $i$ is $v_i(S)=1 $ if $S \cap J_i \neq \emptyset$ and $0$
otherwise.
Namely, the buyers are indifferent between the items in their set (they all have the same valuation for the items in their set, and
a zero valuation for all other items).
The utility of buyer $i$ is quasi-linear, that is, when she receives
items $S$ and pays $p$ her utility is $u_i(S,p)=v_i(S)-p$.
We assume that the subsets $J_i$ are selected uniformly at random
and that $kn/m=O(1)$.\footnote{We require this for ease of analysis.
However, it suffices that the sets are distributed in such a way as
to resemble a uniform or binomial distribution \cite{MRVX12}.}

Our goal is to design a local mechanism that maximizes the social
welfare.
In order to do this, we would like to satisfy as many buyers as
possible, allocating each buyer a single item from her set.
We call this type of auction an $k-$UDUV (unit demand, uniform
value) auction.

Ideally, we would like to find a maximum matching between the buyers
and items, as this will maximize the social welfare. However, it is
not possible to solve the maximum matching problem locally, as the
following theorem shows. This is one important challenge that
the local setting adds to the algorithmic mechanism design.

\begin{theorem}
\label{thm:maxmatch}
 There does not exist an LCA for the maximum
matching problem.
\end{theorem}
 \begin{proof}
To see why it is not always possible to solve the maximum matching
locally, consider the following family of homomorphic graphs:
$\mathcal{G} = \{G_i\}$. All $G_i \in \mathcal{G}$ have $2n$
vertices: $\{v_1, v_2, \ldots, v_{2n}\}$. In each $G_i$, vertices
$v_{-i}$ comprise an (odd) cycle, and vertex $v_i$ is connected to
vertex $v_{i-1}$ (modulo $2n$). Each $G_i$ has a unique maximum
matching. We are given as input a graph $G \in \mathcal{G}$, (i.e.,
we know it is $G_i$ for some $i$, but we don't know the value of
$i$). We would like to know whether the edge $e=(v_1, v_2)$ is in
the maximum matching. Note that the edge $e$ will be in exactly half
of the maximum matchings.

Assume that the graph is either $G_n$ or $G_{n+1}$.  In the distributed model, this implies that the distance between the edge $e=(v_1, v_2)$
and the distinguishing place of $G_{n}$ and $G_{n+1}$ is $n$ edges, which will be a lower bound on the time to detect the correct graph.

In the local computation model, we can write the edges in a random order.
This implies that one needs to query, on average, $n$ edges to distinguish between $G_{n}$ and $G_{n+1}$.

Therefore there cannot exist an LCA for maximum matching.
\end{proof}

\begin{corollary}
\label{corr:bip}
 There does not exist an LCA for the maximum
matching problem in bipartite graphs.
\end{corollary}
 \begin{proof}
The proof is similar to the general case. 
In the  bipartite case, though, in each $G_i$, the vertices $v_1,
v_2, \ldots, v_{i-1},$ $ v_{i+2}, \ldots, v_{2n}$ comprise an
\emph{even} cycle. Vertex $v_{i}$ is connected to vertex $v_{i-1}$
and vertex $v_{i+1}$ is connected to vertex $v_{i+2}$ (modulo $2n$).
Note that the edges $(v_{i-1}, v_{i})$ and $(v_{i+1}, v_{i+2})$ must
be in the maximum matching, because the maximum matching in this
case is of size $n$, meaning all vertices must be matched, including
$v_i$ and $v_{i+1}$.
%
%
\end{proof}

Since there is no local algorithm for computing the maximum
matching\footnote{Actually, we only require a maximum matching in
the special case of bipartite graphs. Corollary \ref{corr:bip} shows
that this is not possible either.}, we will content ourselves with
finding an \emph{approximation} to the maximum matching.

\subsubsection{A $\frac{1}{2}$-approximation to the maximum matching}
 To obtain a  $\frac{1}{2}$-approximation, we use the local greedy matching
algorithm of \cite{MRVX12}, which we denote by $\mcA_{UDUV}$. Algorithm
$\mcA_{UDUV}$ simulates the well-known greedy on-line algorithm - as an
edge arrives, it is added to the matching, if possible.

Our mechanism $\mcM_{UDUV} = (\mcA_{UDUV}, \mcP_{UDUV})$ works as follows.
The mechanism receives from each buyer $i$ a subset $J'_i \subset J$.
For the allocation algorithm $\mcA_{UDUV}$, the mechanism decides on a random order
in which it considers the items.
Specifically, the mechanism assigns each item $j$ a real number $r_j
\sim_u [0,1]$, sampled independently and uniformly. The order of the
items is determined by $r_j$ (higher $r_j$ items are considered first).
Notice that because the $r_j$ are allocated independently,
  buyers cannot influence the order in which the items
as considered.
Given this order induced by $r_j$, the mechanism considers items one
at a time. When item $j$ is considered, if there is some buyer $i$
such that $j\in J'_i$ and buyer $i$ was not allocated any item yet,
then $j$ is allocated to buyer $i$. If there is more than one such
buyer, ties are broken lexicographically.

As the values of all items are identical, any payment scheme
$\mcP_{UDUV}$ that fulfills \emph{voluntary participation}
is adequate, i.e., the payment can be any value in the range
$[0,1]$. For example, charge $p=1/2$ from any buyer that receives an item
and $p=0$ from any buyer that does not receive an item.

In order to show that our mechanism is truthful, we need only to
show that buyers cannot profit by bidding $J'_i \neq J_i$.

\begin{theorem}
\label{thm:fixed_price} In the $k$-UDUV auction, the mechanism
$\mcM_{UDUV} = (\mcA_{UDUV}, \mcP_{UDUV})$  is universally truthful and provides a
$\frac{1}{2}$-approximation to the optimal allocation.
\end{theorem}

\begin{proof}
The proof will be done in two steps. First, we show that for any
$J'_i$, bidding $J_i \cap J'_i$ weakly dominates bidding $J'_i$.
Second, we show that bidding $J_i$ weakly dominates bidding any
$J^*_i \subseteq J_i$.

To show that $J_i \cap J'_i$ weakly dominates bidding $J'_i$, label
the items in $J_i \cap J'_i$ as \emph{good}, and those in $ J'_i
\setminus J_i$ by \emph{bad}. If a good item is allocated to buyer
$i$ when she bids $J'_i$, it will also be allocated to her when
bidding $J_i \cap J'_i$. Therefore the value of buyer $i$ cannot
decrease by bidding $J_i \cap J'_i$, and hence $J_i \cap J'_i$
weakly dominates bidding $J'_i$.

To show that bidding $J_i$ weakly dominates bidding any $J^*_i
\subseteq J_i$, consider the following.
If buyer $i$ does not receive any items when bidding $J^*_i$, the
claim trivially holds. Assume buyer $i$ receives item $j$ when
bidding $J^*_i$. Then, when bidding $J_i$, if she has not received
any item from $J_i \setminus  J^*_i$ before considering item $j$,
then she will receive item $j$. Therefore, if she receives an item
when bidding $J^*_i $, she will also receive an item when bidding
$J_i$, and have the same valuation and utility. This proves that
bidding $J_i$ weakly dominates bidding $J^*_i \subseteq J_i$.


The reasoning that the allocation is a $\frac{1}{2}$-approximation is similar
to the proof of maximal versus maximum matching.
Consider a buyer which is not allocated an item in $\mcA_{UDUV}$ and \emph{is}
allocated an item in the optimal allocation. Her item is allocated
to a unique different buyer in $\mcA_{UDUV}$. This bounds the number of buyers
allocated items in the optimal allocation and not in $\mcA_{UDUV}$ by the
number of buyers that are allocated items in $\mcA_{UDUV}$,  giving
the factor of $\frac{1}{2}$ approximation, and completing the proof of the
theorem.
\end{proof}

The fact that $\mcA_{UDUV}$ is a $(O(\log^{4}{n}),  O(\log^{3}{n}), 1/n)$
- LCA for maximal matching on graph of bounded degree $k$ was shown
in \cite{MRVX12}. The result also implicitly holds for bipartite
graphs in which the degrees are bounded on one side and distributed
binomially on the other. This is exactly the case when buyers are
interested in $k$ items each and the buyers' choices can be seen as
sampled uniformly from the items.  Therefore, we derive the
following theorem:

\begin{theorem}
\label{thm:ubuduv}
 The $k$-UDUV auction has an $(O(\log^{4}{n}), $
$O(\log^{3}{n}), $ $1/n)$ - local computation mechanism which is
universally truthful and provides a $\frac{1}{2}$-approximation to the optimal social welfare.
\end{theorem}

\subsection{Unit demand buyers, uniform-buyer-value}
\label{sec:UDSVB}

We have a set $\mathcal{I}$ of $n$ buyers, and a set $\mathcal{J}$
of $m$ items.  Each buyer $i$ is interested in a set of at most $k$
items, $J_i \subseteq \mathcal{J}$, which is public knowledge, and
has a private valuation, $t_i$ (which represents the value of any
item from $J_i$ to buyer $i$).
Buyer $i$'s valuation for subset $S$ is $v_i(S)=t_i$ if $S\cap
J_i\neq \emptyset$, and $0$ otherwise. The utility of buyer $i$ is
quasi-linear, namely her utility of receiving subset $S$ and paying
$p$ is $u_i(S,p)=v_i(S)-p$.

We can treat this auction as a weighted graph $G=(V,E)$,  in which
$V=\mathcal{J} \cup \mathcal{I}$, and  $E = \cup_i E_i$ where
$E_i=\{(i,j) : j \in J_i\}$. Every edge $e\in E_i$ has weight $w(e) =
t_i$.

We assume that $J_i$ has a uniform or binomial distribution and that
$kn/m=O(1)$.\footnote{As in the previous section, we require this
for ease of analysis. However, it suffices that the sets are
distributed in such a way as to resemble a uniform or binomial
distribution \cite{MRVX12}.}
In addition, we make the simplifying assumption that the buyers are
Bayesian - the valuations $t_i$ are randomly drawn from some prior
(not necessarily known) distribution, that is identical to all
buyers.
%
We call this type of auction an $k-$UDUBV (unit demand, uniform
buyer value) auction.

We require that if buyer $i$ does not receive an item, she pays
nothing. If buyer $i$ receives an item, the mechanism charges her
$p_i(b)$, where $b$ is the bid vector. (Any buyer will receive at
most one item in the allocation of the mechanism).
We would like to ensure that bidding truthfully is a dominant
strategy for all buyers. Hence, we need to show that, for all $b_i$
and $b_{-i}$, we have $u_i(t_i, b_{-i}) \ge u_i(b_i, b_{-i})$.

The allocation algorithm, $\mcA_{UDUBV}$, is as follows.  First, $\mcA_{UDUBV}$ orders
the buyers by their bids. Starting with the buyer with the highest
bid, each buyer $i$ is allocated an item $j_i\in J_i$ such that
$j_i$ has not yet been allocated. If more than one such item exists,
we allocate the (lexicographically) first $j_i \in J_i$. (We assume
the items have lexicographic order.) If there is no such item, then
buyer $i$ is not allocated any item.
We continue until we cannot allocate any
more items.

First, we claim that the resulting allocation is a
$\frac{1}{2}$-approximation.

\begin{claim}
\label{claimproof1} The allocation algorithm $\mcA_{UDUBV}$ provides a
$\frac{1}{2}$-approximation to the optimal allocation, with respect
to the bids $b$.
\end{claim}


\begin{proof}
The proof is similar to the proof that any  maximal matching is a
$2$-approximation to a maximum matching. Regard the auction as a
bipartite graph $G=(U,W,E)$, with $U$ representing the buyers and
$W$ representing the items. There is a weighted edge  between each
buyer $i$ and every item $j_i \in J_i$. The weight of each edge
$e=(i, j_i)$ is the bid of buyer $i$, $b_i$. The optimal allocation
is a maximum weighted matching, while $\mcA_{UDUBV}$  considers the buyers
in the order of their $b_i$'s and finds a maximal matching.

If an edge $e=(i,j_i)$ is added in $\mcA_{UDUBV}$ but not in the optimal
matching, then it is allocated instead of at most $2$ edges in the
optimal matching (an edge $e'$ containing $i$ and an edge $e''$
containing $j_i$). Because $\mcA_{UDUBV}$ considers edges according to
their weights, we know that $w(e)\geq w(e')$ and $w(e) \geq w(e'')$.
Therefore $2w(e) \geq w(e')+ w(e'')$ and so the ratio between
$\mcA_{UDUBV}$ and the optimal allocation is at least $\frac{1}{2}$.
\end{proof}

We now need to specify the payment mechanism. To calculate buyer
$i$'s payment when she receives an item, we run $\mcA_{UDUBV}$ without
buyer $i$. Buyer $i$ pays the smallest value for which any of her
items is sold when the auction is run without her. (This is
exactly the minimal value of $b_i$ which would still gain her an
item). We label buyer $i$'s payment by $p_i$, hence,  the payments
are $\mcP_{UDUBV}=\{p_1, \ldots p_n\}$.

\begin{claim}
\label{claim:UDUBVB} In mechanism $\mcM_{UDUBV} = (\mcA_{UDUBV}, \mcP_{UDUBV})$, for
all buyers $i$ and all $b_i$, bidding $t_i$ weakly dominates bidding $b_i$.
\end{claim}

\begin{proof}
We will show that, fixing the bids of all other buyers at $b_{-i}$,
\begin{enumerate}
\item Buyer $i$ has no incentive to over-bid, i.e., bid $b_i>t_i$ \label{item:over-bid}.
\item Buyer $i$ has no incentive to under-bid i.e., bid $b_i<t_i$ \label{item:under-bid}.
\end{enumerate}

To prove (\ref{item:over-bid}), we notice that if buyer $i$ receives
an item, then she has no incentive to bid higher, as she has no
preference between items. Furthermore, bidding higher cannot change
her payment, as her payment is independent of her bid. If she does
not receive an item, then  $p_i \geq b_i (=t_i)$\footnote{As
specified, if buyer $i$ does not receive an item, she pays $0$.
However, if the mechanism were to compute the payment, i.e., run the
mechanism without her, the payment would be $p_i \geq b_i (=t_i)$.},
and so if she bids more, she might receive an item, but will have to
pay at least $t_i$ if she does, which will result in a non-positive
utility.

To prove (\ref{item:under-bid}),  we notice that if buyer $i$ does
not receive an item, she cannot obtain an item by bidding lower,
because the algorithm allocates first to higher bids. If she is
allocated an item, then bidding lower will not make a difference,
unless she bids under  $p_i$, in which case she will not receive any
item, and hence have zero utility.
\end{proof}

Claims~\ref{claimproof1} and \ref{claim:UDUBVB} imply the following.

\begin{theorem}
\label{thm:matching_general} The mechanism $\mcM_{UDUBV}$ is universally
truthful and provides a $\frac{1}{2}$-approximation to the optimal social welfare.
\end{theorem}

Algorithm $\mcA_{UDUBV}$ is a  $(O(\log^{4}{n}),  O(\log^{3}{n}), 1/n)$ -
LCA for maximal matching on graph of bounded degree $k$, by
\cite{MRVX12}. Notice, however, that we need to run $\mcA_{UDUBV}$ once
for calculating the allocation, and $k$ more times for calculating
the payment. Hence, we have the following.
{
\renewcommand{\thetheorem}{\ref{thm:ubuduv}}
\begin{theorem}
(2)
There is an $(O(\log^{4}{n}), $ $O(\log^{3}{n}),
$ $1/n)$ - local mechanism for $k-$UDUBV auction which is
universally truthful and provides a $\frac{1}{2}$-approximation to the optimal social welfare.
\end{theorem}
\addtocounter{theorem}{-1}
}


\section{Single minded bidders}
\label{sec:SMB1}

We extend the results of Section \ref{sec:UDSVB} to the case
of combinatorial auctions with single-minded bidders: There is a set
$\mathcal{I}$ of $n$ buyers, and a set $\mathcal{J}$ of $m$ items.
Each buyer $i$ is interested in a set of at most $k$ items, $J_i
\subseteq \mathcal{J}$, which is public knowledge, and has a private
valuation, $t_i$, which represents the value of the entire subset
$J_i$ to buyer $i$. Buyer $i$'s valuation for subset $S$ is
$v_i(S)=t_i$ if $J_i \subseteq S$, and $0$ otherwise. The utility of
buyer $i$ is quasi-linear, namely her utility of receiving subset
$S$ and paying $p$ is $u_i(S,p)=v_i(S)-p$.

As in Subsection \ref{sec:UDSVB}, we assume that $J_i$ has a uniform
or binomial distribution, and that the valuations $t_i$ are randomly
drawn from some prior (not necessarily known) distribution,and
$kn/m=O(1)$.

The allocation algorithm, $\mcA_{kSMB}$, is as follows.  First, $\mcA_{kSMB}$ orders
the buyers by their bids. Starting with the buyer with the highest
bid, each buyer $i$ is allocated subset $J_i$ such that no item
$j_i \in J_i$ has been allocated yet.
We continue until we cannot allocate any more subsets.

\begin{claim}
\label{claim:smb1app} The allocation algorithm $\mcA_{kSMB}$ provides a
$\frac{1}{k}$-approximation to the optimal allocation, with respect to the
values $b$.
\end{claim}

\begin{proof}
Compare the allocation of Algorithm $\mcA_{kSMB}$, $J^*$,  to the optimal
allocation, $OPT$. Each set $J \in J^*$ is chosen by  $\mcA_{kSMB}$
instead of at most $k$ sets in $OPT$, but its weight is greater than
each of their weights, because $\mcA_{kSMB}$ is a greedy algorithm.
\end{proof}

The payment scheme is as follows. To calculate buyer
$i$'s payment when she receives an item, we run $\mcA_{kSMB}$ without
buyer $i$. Buyer $i$ pays the highest value of the allocated sets $J_x$ for which $J_i \cap J_x \neq \emptyset$. (This is
exactly the minimal value of $b_i$ which would still gain her an
item). We label buyer $i$'s payment by $p_i$, and let $\mcP_{kSMB}=\{p_1, \ldots p_n\}$.

\begin{claim}
\label{claim:SMB1truth} In mechanism $\mcM_{kSMB} = (\mcA_{kSMB}, \mcP_{kSMB})$, for
all buyers $i$ and all $b_i$, bidding $t_i$ weakly dominates bidding $b_i$.
\end{claim}
The proof is similar to the proof of Claim \ref{claim:UDUBVB} and is omitted.

Combining Claims \ref{claim:smb1app} and \ref{claim:SMB1truth}, we get

\begin{theorem}
\label{thm:smb1}
There is an $(O(\log^{4}{n}), $ $O(\log^{3}{n}),
$ $1/n)$ - local mechanism for combinatorial auctions with known $k-$single minded bidders (where the sets are sampled uniformly at random), which is
universally truthful and provides a $\frac{1}{k}$-approximation to the optimal social welfare.
\end{theorem}

\section{Random serial dictatorship}
\label{sec:rsd}
We would like to allocate $n$ houses to $n$ agents.  We assume that each agent is interested in a constant number of houses, $d$, and that the preferences are drawn from the uniform distribution. Each agent $i$ has a complete preference relation $R_i$ over the $d$ houses. In the random serial dictatorship algorithm, a permutation over the agents is generated, and then each agent chooses her most preferred house out of the unallocated houses.
We would like to simulate this algorithm locally:
Each agent $i$  is allocated, uniformly and at random, an integer $r_i \in [n^4]$. $r$ determines the permutation: if $r_i<r_j$ then agent $i$ makes a choice before agent $j$. If $r_i=r_j$, the tie is broken lexicographically. Agent $i$ checks, for each of its $d$ housing choices in the order induced by $R_i$, whether it has been allocated already. If it has not, $i$ chooses the house. For $i$ to check whether or not house $j$ has already been allocated, it needs to check whether any of the agents interested in $j$ had already chosen it. $i$ therefore needs to recursively check the allocation of all agents that arrived before her, on which her allocation depends.
As in \cite{ARVX11, MRVX12}, we model this using a query tree. We recall the following lemma from \cite{MRVX12}:
\begin{lemma}[\cite{MRVX12}]
\label{lemma:bipartite}
Let $G=(\{V,U\},E)$ be a bipartite graph, $|V| = n$ and $|U| = m$ and $n= cm$ for some constant $c\geq 1$, such that for each
vertex $v \in V$ there are $d$ edges chosen independently and at random between $v$ and $U$. For any constant $\alpha$ there is a constant
$C$ which depends only on $d$ and $\alpha$ such that 
\begin{center}
$\Pr[|T|<C \log{n}]>1-1/n^{\alpha}$,
\end{center}
where the probability is taken over all of the possible permutations $\pi \in \Pi$ of the vertices of $G$, and $T$ is a random query tree in $G$ under $\pi$.
\end{lemma}

This implies that in order to reply to  any query, we will require more that $O(\log{n})$ queries to the graph with probability at most $1/n^2$, which implies the following theorem:

\begin{theorem}
\label{thm:rsd}
Let $k$ be some constant integer $k>0$. Consider a house allocation problem with $n$ agents and $n$ houses, and let each agent preference list length be bounded by $k$, where each list is drawn uniformly at random from the set of all possible lists of length $k$. Then there is an $(O(\log^4{n}), O(\log^3{n}), 1/n)$ - LCA which simulates, and whose output is identical to,  the Random Serial Dictatorship allocation algorithm.
\end{theorem}

\section{Open questions}
\label{sec:conclusions}

We provide several interesting open questions.
\smallskip

\paragraph{Open question 1} How restrictive is the local computation setting? It remains an interesting open question in general, and specifically for the load balancing problem, while the intractability of the non-local problem carries over to the local computation setting, there exists a truthful PTAS for the non-local setting. How well can we approximate the optimal solution using LCAs? Furthermore, it would be interesting to quantify the added difficulty imposed by the requirement of truthfulness. 
\smallskip

\paragraph{Open question 2}
An intriguing open problem is whether there exists a truthful local
mechanism in the unit-demand combinatorial auction setting, where both the sets and the valuations
are private information. We conjecture that this is impossible, especially
considering a similar result of  \cite{FLSS11}, albeit
in a different setting (auctions with budgets).

\smallskip

\paragraph{Open question 3}
In this paper, we show an LCA which provides a $1/2$-approximation to the optimal solution for unit-demand buyers, using a maximal matching LCA. \cite{MV13},  showed that there exists a $(1-\eps)$-approximation algorithm to maximum matching. Unfortunately, that algorithm does not yield a monotonic allocation. Is there an LCA which provides a better than $1/2$-approximation that can be transformed to a local computation mechanism?

\smallskip

\paragraph{Open question 4}
If we view the $k$-single minded bidder combinatorial auction as a  hypergraph $H=(V,E)$, where each item
is represented by a vertex $v \in V$, and each player by a weighted
hyper-edge $e \in E$, the problem is reduced to maximal weighted
independent set in $k$-regular hypergraphs, (equivalently  weighted
$k$-set  packing). \cite{CH01} show how
to approximate weighted $k$-set packing to within $2(k+1)/3$, via
local improvements. It would be interesting to see if it would be
possible to apply of their techniques to obtain a local mechanism
with a better approximation ratio.

\smallskip

\paragraph{Open question 5}
In this paper, we often need to make the assumption that the demands of the bidders are drawn from some distribution in order to ensure the locality of our mechanisms. This is quite a strong assumption, and it would be interesting to see how far it can be relaxed. 

\smallskip
\paragraph{Acknowledgements}
We would like to thank Amos Fiat, Alon Naor and Amit Weinstein  for their useful input.

\bibliographystyle{alpha}
\bibliography{Vardi_PhD_Bibliography}



\appendix
%
%

\section{Proof of Lemma~\ref{lemma:loglog}}
\label{appendix:lglgn}
{
\renewcommand{\thetheorem}{\ref{lemma:loglog}}
\begin{lemma}
The allocation algorithm $\mcA_{RLMS}$ provides an $O(\log{\log{n}})$-approximation to the optimal allocation.
\end{lemma}
\addtocounter{theorem}{-1}
}

We prove the theorem for the case $d=2$ (each job can be assigned to one of $2$ machines). The proof is easily expandable for the case of $d>2$. For the proof (not the algorithm), we regard each machine $i$ of capacity $c_i$ as having $c_i$ slots of capacity $1$. Before presenting the proof we need several definitions:

The \emph{load vector} of an allocation of jobs into $n$ machines is $L= (\ell_1, \ldots , \ell_n)$, where $\ell_i = h_i$ is the load of machine $i$. The \emph{normalized load vector} $\bar{L}$ consists of the members of $L$ in non-increasing order (where the order among machines with the same load is arbitrary). For the case of non-uniform machines of capacities $c_1, \ldots c_n$, and total capacity $C = \sum_{i=1}^{n} c_i$, we define the \emph{slot-load vector} $S=(h_{1,1}, \ldots h_{1,c_1}, h_{2,1}, \ldots h_{2,c_2}, \ldots h_{n,1}, \ldots h_{n,c_n})$, where if machine $i$ contains $r$ jobs, the first $r$ mod $c$ slots will have $\lceil r/c \rceil$ jobs, and the remaining slots will have $\lfloor r/c \rfloor$ jobs. If a machine has an uneven allocation of jobs, we call the slots with more jobs \emph{heavy}, and the slots with less jobs \emph{light}. If the load on a machine is an integer (all of the slots of the machine have an identical number of jobs assigned to them), we call all the slots \emph{light}. When we add a job to a machine, we add the job to one of the light slots, arbitrarily. The \emph{normalized slot load vector} $\bar{S}$ is $S$ sorted in non-increasing order (slots of the same machine may be separated in $\bar{S}$). We add a subscript $t$ to these vectors, i.e., $L_t$, $\bar{L}_t$, $S_t$ and $\bar{S}_t$ to indicate the vector after the allocation of the $t$-th job.

\begin{definition}
[Majorization, $\maj$] We say that a vector $P=(p_1, \ldots, p_n)$ \emph{majorizes} vector $Q=(q_1, \ldots, q_m)$ (denoted $P \maj Q$) if and only if for all $1 \leq k \leq min(m,n)$,
\begin{equation*}
\dsum_{i=1}^k \bar{p_i} \geq \dsum_{i=1}^k \bar{q_i},
\end{equation*}
where $\bar{p}_i$ and $\bar{q}_i$ are the $i$-th entries of the normalized vectors $\bar{P}$ and $\bar{Q}$.
\end{definition}
For $n \in \N$, let $[n]$ denote $\{1,\ldots,n\}$.
\begin{definition} [System Majorization]
Let $A$ and $B$ be two processes allocating $m$ jobs into machines with the same total capacity $n$. Let $\tau=(\tau_1 \ldots \tau_{2m})$, $\tau_i \in [n]$ be a vector representing the (slot) choices of the $m$ jobs ($\tau_{2i-1}$ and $\tau_{2i}$ are the choices of the $i$-th job). Let $S^A(\tau)$ and $S^B(\tau)$ be the slot load vectors using $A$ and $B$ respectively with the random choices specified by $\tau$. Then we say
\begin{enumerate}
\item  $A$ \emph{majorizes} $B$ (denoted by the overloaded notation $A \maj B$) if there is a bijection $f:[n]^{2m}\rightarrow [n]^{2m}$ such that for all possible random choices $\tau \in [n]^{2m}$,we have
\begin{equation*}
L^A(\tau) \maj L^B(f(\tau))
\end{equation*}
\item The maximum load of $A$ \emph{majorizes} the maximum load of $B$ (denoted by $A \maj_m B$) if there is a bijection $f:[n]^{2m}\rightarrow [n]^{2m}$ such that for all possible random choices $\tau \in [n]^{2m}$,we have
\begin{equation*}
\ell_1^A(\tau) \geq \ell_1^B(f(\tau)),
\end{equation*}
where $\ell_1^A(\tau)$ and $\ell_1^B(f(\tau))$ are the loads of the most loaded bins in $A$ and $B$ respectively with the random choices specified by $\tau$ and $f(\tau)$ respectively. 
\end{enumerate}
\end{definition}

\begin{obs}
\label{obs:rightarrow}
 $A \maj B \Rightarrow A \maj_m B$. 
\end{obs}

 We now turn to the proof of Lemma \ref{lemma:loglog}.

First, notice that if we have an system of $m$ identical machines, each of capacity $1$, both the unmodified Greedy algorithm and the allocation algorithm $\mcA_{RLMS}$ will behave in exactly the same way - the load and the $\lfloor$load$\rfloor$  are the same if the capacity is $1$. From \cite{ABK+99}, we know that the maximal load on any machine when allocating $m=n$ jobs (to $n$ machines with capacity $1$) with the Greedy algorithm, is $\Theta(\log{\log{n}})$. Therefore, the maximal load when allocating $m=m$ jobs with  $\mcA_{RLMS}$ is also $\Theta(\log{\log{n}})$ in this setting. We would like to show that the maximal load of a system with non-uniform machines of total capacity $C$ is majorized by the maximal load of a system with $C$ machines of capacity $1$, when the allocating algorithm is $\mcA_{RLMS}$. We will show that the first system majorizes the second, and deduce the required result from Observation \ref{obs:rightarrow}.

We restate Claim $2.4$ of \cite{Wie07}:
\begin{claim}[\cite{Wie07}]
\label{claim:wie}
Let $P$ and $Q$ be two normalized integer vectors such that $P \maj Q$. If $i \leq j$ then $P + e_i \maj Q+e_j$ where $e_i$ is the $i$-th unit vector and $P + e_i$ and $Q+e_j$ are normalized.
\end{claim}

\begin{lemma}
For allocation algorithm $\mcA_{RLMS}$, let $A$ be a system with non-uniform machines of total capacity $C$, and $B$ be a system with $C$ uniform machines of capacity $1$ each.
Then $B \maj A$.
\end{lemma}
\begin{proof}
We use the slot load vectors of systems $A$ and $B$ (in $B$ the load vector and slot load vector are identical), and show that $S^B(f(\tau)) \maj S^A(\tau)$. The bijection is such that the jobs in both processes choose the same $k_1< k_2 \in \{1,\ldots, C\}$ in the normalized slot load vectors, and the choice corresponds to machines $k_1, k_2$ in $B$ and the machines associated with those specific slots in system $A$.
We use induction: for $t=0$, the claim is trivially true.

From the inductive hypothesis, before the allocation of the  $t$-th job, $S^B_{t-1}(f(\tau)) \maj S^A_{t-1}(\tau)$. In system $B$, the $t$-th job goes to machine $k_2$.  In system $A$, if the $\lfloor$load$\rfloor$ of the machine of $k_1$ is greater than that of the machine of $k_2$, the job goes to $k_2$ if $k_2$ is a light slot, or to a slot to the right of $k_2$ (a lighter slot of the same machine), if $k_2$ is a heavy slot. If the $\lfloor$loads$\rfloor$ of the machines of $k_1$ and $k_2$ are the same, again, the job goes to $k_2$ if $k_2$ is a light slot, or to a slot to the right of $k_2$ (again, a lighter slot of the same machine), if $k_2$ is a heavy slot.
In all cases, by Claim \ref{claim:wie}, it follows that $S^B(f(\tau)) \maj S^A(\tau)$.
\end{proof}

\end{document}